\documentclass[a4paper, reqno]{amsart}
\usepackage{a4,wasysym}
\usepackage{amsmath,graphicx, amssymb,color,ulem,bbm}
\usepackage{epsf, subfigure, verbatim}
\usepackage{float}
\usepackage{epsfig}

\newtheorem{theorem}{Theorem}

\newtheorem{proposition}[theorem]{Proposition}

\newcommand{\be}{\begin{equation}}
\newcommand{\ee}{\end{equation}}

\newcommand{\W}{{\it Wolbachia}}

\newtheorem{thm}{\bf Theorem}[section]

\newcommand{\Wolb}{{\it Wolbachia }}

\begin{document}

\title[Control of West Nile virus by {\em Wolbachia}]{Modelling {\em Wolbachia} infection in a sex-structured mosquito population carrying West Nile virus}

\author[J\'{o}zsef Z. Farkas]{J\'{o}zsef Z. Farkas}
\address{J\'{o}zsef Z. Farkas, Division of Computing Science and Mathematics, University of Stirling, Stirling, FK9 4LA, United Kingdom }
\email{jozsef.farkas@stir.ac.uk}

\author[Stephen A. Gourley]{Stephen A. Gourley}
\address{Stephen A. Gourley, Department of Mathematics, University of Surrey, Guildford, Surrey, GU2 7XH, United Kingdom}
\email{s.gourley@surrey.ac.uk}

\author[Rongsong Liu]{Rongsong Liu}
\address{Rongsong Liu, Department of Mathematics and Department of Zoology and Physiology, University of Wyoming, Laramie, WY 82071, USA}
\email{Rongsong.Liu@uwyo.edu}

\author[Abdul-Aziz Yakubu]{Abdul-Aziz Yakubu}
\address{Abdul-Aziz Yakubu,  Department of Mathematics, Howard University, Washington, DC 20059, USA}
\email{ayakubu@howard.edu}

\date{\today}

\begin{abstract}
\Wolb is possibly the most studied reproductive parasite of arthropod species. It appears to be a promising candidate for 
biocontrol of some mosquito borne diseases. We begin by developing a sex-structured model for a \Wolb infected mosquito population. 
Our model incorporates the key effects of \Wolb infection including cytoplasmic incompatibility and male killing. 
We also allow the possibility of reduced reproductive output, incomplete maternal transmission, and different mortality rates for uninfected/infected male/female individuals. We study the existence and local stability of equilibria, including the biologically relevant and interesting boundary equilibria. For some biologically relevant parameter regimes there may be multiple coexistence steady states including, very importantly, a coexistence steady state in which \Wolb infected individuals dominate. We also extend the model to incorporate West Nile virus (WNv) dynamics, using an SEI modelling approach. Recent evidence suggests that a particular strain of \Wolb infection significantly reduces WNv replication in {\it Aedes aegypti}. We model this via increased time spent in the WNv-exposed compartment for \Wolb infected female mosquitoes. 
A basic reproduction number $R_0$ is computed for the WNv infection. Our results suggest that, if the mosquito population consists mainly of \Wolb infected individuals, WNv eradication is likely if WNv replication in \Wolb infected individuals is sufficiently reduced.
\end{abstract}
\maketitle

\section{Introduction}

{\em Wolbachia} is a maternally transmitted intracellular symbiont, 
and it is the most common reproductive parasite infecting a significant 
proportion of insect species, see e.g. \cite{ONeill,Werren1997}. 
{\em Wolbachia} typically inhibits testes and ovaries of its host, 
and it is also present in its host's eggs. It interferes with its host's 
reproductive mechanism in a remarkable fashion. This allows \Wolb to successfully establish itself in 
a number of arthropod species. Well-known effects of {\em Wolbachia} infections include cytoplasmic incompatibility (CI for short) and  
feminization of genetic males also known as male killing (MK for short), 
see e.g. \cite{Caspari1959,Hoffmann1997,Telschow2005,Telschow2005b}. 
Another important well-known effect of {\em Wolbachia} infections is the inducement of parthenogenesis, see e.g. \cite{Engelstadter2004,Stouthamer1997}. All of these contribute to the fact that the mathematical modelling of \Wolb infection dynamics is both interesting and challenging. 

In recent decades a substantial number of mathematical modelling approaches have been 
applied to model different types of \Wolb infections in a variety of arthropod species.  
Perhaps most frequently researchers have been focusing on the development of mathematical models for 
\Wolb infections in mosquito species. Many of the earlier models took the form of discrete time matrix models, written 
for population frequencies, see e.g.  \cite{Turelli1994,Vautrin2007}, and the references therein. 
Using frequency-type models a number of researchers investigated the possibility of coexistence of multiple \Wolb strains, 
each of which exhibits different types of the reproductive mechanisms mentioned earlier, see e.g. \cite{Engelstadter2004,FHin2,Keeling2003,Vautrin2007}. 
Among others, \Wolb strains have been investigated as a potential biological control tool to eradicate mosquito borne diseases. 
Originally the focus has been on \Wolb strains that induce life-shortening of their hosts. This is because for many vector borne diseases 
only older mosquitoes are of interest from the point of view of disease transmission. 
Therefore the use of (discrete) age-structured population models has become increasingly prevalent, see e.g. \cite{Rasgon2004} and the references therein. Fairly recently, in \cite{McMeniman2009} the results of laboratory experiments were reported envisaging a successful introduction 
of a life-shortening \Wolb strain in the mosquito species {\em Aedes aegypti}. 
In \cite{McMeniman2009} three key factors, namely, strong CI, low fitness cost and high maternal transmission rate, were identified as drivers of a successful introduction of the new \Wolb strain into an {\em Aedes} population. 
To this end researchers have developed and analysed continuous age-structured population models for \Wolb infection dynamics, 
which take the form of partial differential equations, see \cite{FHin2}; which can often be recast as delay equations, see e.g. \cite{Hancock,Hancock2}.

In recent years there have been substantial modelling
  efforts to theoretically investigate the potential
  of biological control tools for limiting the impact of
  mosquito borne diseases. It is now widely recognised that biological
  control represents a viable alternative to established methods
  such as the use of insecticides and bed nets. Among others, the sterile insect
  technique has been investigated in the recent papers
  \cite{Dumont2012,Li2011,Li2015}. More recently, it was reported
that particular strains of \Wolb (completely or almost completely)
block dengue virus replication   
inside the mosquito hosts, see for example \cite{Blagrove,Hoffmann,Walker2011}. 
To this end Hughes and Britton~\cite{Hughes_2013} developed a
mathematical model for \Wolb infection as a potential control tool for
dengue fever.  
Their work suggests that \Wolb may be effective as such a control measure in areas where the basic reproduction number $R_0$ is not too large.  These recent results underpin the possibility that \Wolb may be a promising candidate for biocontrol of mosquito borne diseases, in general. 
Besides dengue, West Nile virus (WNv) is another well-known mosquito borne disease of current interest. 
WNv infection cycles between mosquitoes (especially \textit{Culex} species) and a number of species, particularly birds. Some infected
birds develop high levels of virus in their bloodstream and mosquitoes can become infected by biting these infectious birds. 
After about a week, infected mosquitoes can transmit the virus to susceptible birds.  
Mosquitoes infected with West Nile virus also bite and infect people, horses, and other mammals. 
However, humans, horses, and other mammals are `dead end' hosts. This virus was first isolated in
the West Nile region of Uganda, and since then has spread rapidly, for
example in North America during the past 12 years. Since there is no
vaccine available, the emphasis has been mainly on controlling the
vector mosquito species. Some recent experiments, see
\cite{Hussain2013}, have confirmed that replication of the virus in
orally fed mosquitoes was largely inhibited in the wMelPop strain of
\textit{Wolbachia}. Interestingly, in a recent paper,
Dodson~et~al~\cite{Dodson2014} demonstrated in laboratory experiments
  that the {\it w}AlbB {\it Wolbachia} strain in fact enhances WNv
  infection rates in the mosquito species {\it Culex
    transalis}. However, in \cite{Dodson2014} the {\it Wolbachia} was
  not a stable maternally inherited infection, but rather they
  infected transiently somatic mosquito tissues, and hence the {\it
    w}AlbB infection did not induce significant immune response in the
  mosquitoes. This is probably key to their findings. Here we will
  focus on modelling a maternally inherited {\it Wolbachia} infection
  in a population model, which hypothesizes a large number of
  successive generations. Nevertheless, the findings in
  \cite{Dodson2014} underpin the importance of {\it Wolbachia}
  research in general and highlight the importance of contrasting the findings
  of new theoretical, laboratory and field investigations.

In this work we introduce sex-structured models for \Wolb infection dynamics in a mosquito population. 
This will allow us to incorporate and study the well-known effects of CI and MK of particular \Wolb infections, simultaneously. 
First we will treat a model which only involves the mosquito population itself. Then we will 
use this model as a basis for a much more complex scenario incorporating WNv dynamics in a \Wolb infected mosquito population. 
The full WNv model will naturally include the bird population, too.

\section{Model for a \Wolb infected mosquito population without WNv}

\subsection{Model derivation}
We start by introducing a model for a \Wolb infection in a sex-structured mosquito population, incorporating sex-structure using 
a well established approach originally due to Kendall
\cite{Kendall}. More recent papers of
  Hadeler~\cite{Hadeler2012} and Hadeler~et~al~\cite{Hadeler1988} derive
  and discuss sex-structured pair formation models in depth.
We only model (explicitly) the adult population of mosquitoes. Our model allows us to take into account the well-known effects of 
cytoplasmic incompatibility (CI), incomplete maternal transmission, fertility cost of the \Wolb infection to reproductive output, 
and male killing (MK), at the same time. We note that it was shown in \cite{Engelstadter2004} that a stable
coexistence of MK and CI inducing {\it Wolbachia} strains is possible, in principle. 
Introduction of male killing {\it Wolbachia} strains in vector populations may have a significant effect on the disease
dynamics, as typically only female mosquitoes are transmitting the disease. Also note that according to \cite{Walker2011}, those \Wolb strains which cause greater disruption, as in the case of dengue transmission, confer greater fitness costs to the mosquitoes. 
This may well be the case for West Nile virus, hence we account for the reduced reproductive output in our model.  

We deduce our starting model from basic principles. In particular, first we deduce mating rules arising at the individual level. 
Starting with the adult population of size $N$, we construct a random mating graph. 
This is a bipartite graph, not necessarily complete, in which each vertex has degree at most one. The vertices represent male and female individuals and edges represent realized matings. Let us denote by $M,M_w,F,F_w$ the numbers of un/infected males and females, respectively. For every adult mating pair, offspring is created according to the following rules. Below, $m,f$ and $m_w,f_w$ denote uninfected/infected male/female individual, respectively. The parameter $\beta$ models the reduced reproductive output of \Wolb infected females, $\tau$ measures maternal transmission in the sense that it is the probability that a \Wolb infected mother passes on the infection to its offspring, $q$ measures CI in the sense that when a \Wolb infected male mates with an uninfected female, $q$ is the probability that there is no viable offspring. Finally, $\gamma$ measures MK in the sense that it is the probability that a \Wolb infected male larva dies during its development. A complete list of parameter values will be given later on. With this notation the mating rules are described below.
\begin{enumerate}
\item $m \times f$ : create one pair of the same type  $(m,f)$. 
\item $m \times f_w$ : with probability $\beta$, create no
  offspring. This reflects the fecundity reduction due to the
  Wolbachia infection. In the complementary case, with probability
  $(1-\beta)\tau(1-\gamma)$, create  a new pair $(m_w, f_w)$, at the
  same time with probability $(1-\beta)\tau\gamma$ create $(0,f_w)$,
  i.e. a female only brood. This accounts for male killing (MK). With probability $(1-\beta)(1-\tau)$ create a new pair $(m,f)$.
\item $m_w\times f_w$ : same as above.
\item $m_w\times f$: with probability $q$, create no offspring. This is the effect of cytoplasmic incompatibility (CI). With probability $1-q$, create a new pair $(m, f)$.
\end{enumerate}
Notice that, in contrast to \cite{FHin2}, the sex ratio at birth will not be $1:1$, it is distorted due to male killing. Also we allow
different  mortality rates for males and females, in general. Therefore, even in the case when there is no male killing, the
sex ratio would be distorted, in general. Also, we assume that any offspring resulting from CI crossing is uninfected.  

We apply the mating rules described above to construct the birth function in our model. If the population sizes in the four compartments are denoted by $M,M_w,F,F_w$, respectively, then the total number of possible matings is $(M+M_w)(F+F_w)$. 
The total number of matings for example between uninfected males and infected females is $MF_w$. Hence the probability that a given mating of type $m\times f_w$ takes place is $\frac{MF_w}{(M+M_w)(F+F_w)}$. To compute the total number of matings per unit time we follow the harmonic mean birth function approach from \cite{Keyfitz_1972}. Accordingly, the total number of matings is proportional to
\begin{align}
& M\left(\frac{F+F_w}{M+M_w+F+F_w}\right)+M_w\left(\frac{F+F_w}{M+M_w+F+F_w}\right) \nonumber \\
& +F\left(\frac{M+M_w}{M+M_w+F+F_w}\right)+F_w\left(\frac{M+M_w}{M+M_w+F+F_w}\right)=2\frac{(M+M_w)(F+F_w)}{M+M_w+F+F_w}.
\end{align}
Hence the birth rate of offspring arising for example from the mating between an uninfected male and an infected female 
is proportional to $\frac{MF_w}{M+M_w+F+F_w}$. We also naturally
assume that there is competition between female individuals  
for finding an appropriate water reservoir to lay eggs. This is taken
into account via a function $\lambda(F_{total})$  which we assume (at least in the first instance)
to be a monotonically decreasing function of the total number of
females $F_{total}$, to allow for this competition for nesting
places. Though $\lambda(F_{total})$ is decreasing, it may approach a positive
limit as $F_{total}\rightarrow\infty$. This is to allow for the fact
that gravid females that cannot find a place to lay their eggs may
destroy eggs previously laid by others, and lay theirs instead. Thus
the overall egg-laying rate should approach a positive limit as
$F_{total}\rightarrow\infty$, and therefore we assume that $\lambda(F_{total})$
is a decreasing function such that $\lambda(\infty)>0$. Based on
the individual mating rules explained earlier, our model reads as
follows.   
\begin{align}
M^{\prime}(t) & = - \mu_m M + \frac{\lambda(F_{total})}{N} (MF + (1-\beta)(1-\tau)(MF_w + M_wF_w) + (1-q)M_wF),\label{eqM} \\
F^{\prime}(t) & = - \mu_f F + \frac{\lambda(F_{total})}{N} (MF + (1-\beta)(1-\tau)(MF_w + M_wF_w) + (1-q)M_wF), \label{eqF} \\
M^{\prime}_w(t) &= - \mu_{mw} M_w + \frac{\lambda(F_{total})}{N} (1-\beta)\tau(1-\gamma)(MF_w + M_wF_w), \label{eqMw} \\
F^{\prime}_w(t) &= - \mu_{fw} F_w + \frac{\lambda(F_{total})}{N} (1-\beta)\tau(MF_w + M_wF_w). \label{eqFw}
\end{align}
A complete list of the variables, parameters and coefficient functions appearing in model \eqref{eqM}-\eqref{eqFw} is given below. 
\begin{itemize}
\item $M:$ number of uninfected male mosquitoes.
\item $F:$ number of uninfected female mosquitoes.
\item $M_w:$ number of \Wolb infected male mosquitoes.
\item $F_w$: number of \Wolb infected female mosquitoes.
\item $M_{total} = M+M_w$, total number of male mosquitoes.
\item $F_{total} = F + F_w$, total number of female mosquitoes.
\item $N=M_{total}+F_{total}$, total number of mosquitoes.
\item $\beta$: reduction in reproductive output of \Wolb infected females.
\item $\tau$: maternal transmission probability for \Wolb infection.
\item $q$: probability of cytoplasmic incompatibility (CI).
\item $\gamma$: probability of male killing (MK) induced by \Wolb infection.
\item $\lambda(F_{total})$: average egg laying rate, which depends on the total number of female mosquitoes.
\item $\mu_m$: per-capita mortality rate of uninfected male mosquitoes.
\item $\mu_f$: per-capita mortality rate of uninfected female mosquitoes.
\item $\mu_{mw}$: per-capita mortality rate of \Wolb infected male mosquitoes.
\item $\mu_{fw}$: per-capita mortality rate of \Wolb infected female mosquitoes.
\end{itemize}
Model \eqref{eqM}-\eqref{eqFw} is our starting point for a study of
the \Wolb infection dynamics in a sex-structured mosquito population.  
Later, we will expand this model by introducing WNv infection. 

\subsection{Positivity and boundedness}
First we begin by establishing positivity and boundedness of solutions of model \eqref{eqM}-\eqref{eqFw}. 
\begin{proposition}
Assume that $\lambda$ is a monotone decreasing function such that 
\begin{equation}
\displaystyle\lim_{F_{total}\to\infty}\lambda(F_{total})=\lambda_{min}>0, \quad \lambda(0)>\min\{\mu_f,\mu_{fw}\}, \quad \lambda_{min}<\min\{\mu_f,\mu_{fw}\}, \label{lambda-cond}
\end{equation}
hold. Then, the variables $(M,F,M_w,F_w)$ satisfying equations \eqref{eqM}-\eqref{eqFw} remain non-negative if they are non-negative initially, and they remain bounded for all times. 
\end{proposition}
\begin{proof}
First note that the solution variables remain non-negative for all time; this follows from results in \cite{Smith95}. 
Adding equations \eqref{eqF} and \eqref{eqFw}, and noticing that $\beta \in [0,1]$, $q \in [0,1]$,  we  have
\begin{align}
F_{total}^{\prime} \leq & \displaystyle  -\min\{\mu_f, \mu_{fw}\} F_{total} + \frac{\lambda(F_{total})}{M_{total} + F_{total}}M_{total}F_{total}\nonumber \\
\leq & \displaystyle -\min\{\mu_f, \mu_{fw}\} F_{total} + \frac{\lambda(F_{total})}{M_{total} }M_{total}F_{total}\nonumber \\
= & \left(-\min\{\mu_f, \mu_{fw}\} + \lambda(F_{total})\right) F_{total}.
\end{align}
Therefore,
$$
\limsup_{t\rightarrow \infty} F_{total}(t) \leq \bar{F}
$$
where $\bar{F}$ is such that $\lambda(\bar{F}) = \min\{\mu_f,\mu_{fw}\}$. Note $\bar{F}$ exists since we assumed $\lambda$ is monotone decreasing and satisfies \eqref{lambda-cond}.

Since $F_{total}$ remains bounded it follows that $M_{total}$ is bounded as well, because adding \eqref{eqM} and \eqref{eqMw} we have
\begin{align}
M_{total}^{\prime} \leq & \displaystyle  -\min\{\mu_m, \mu_{mw}\} M_{total} + \frac{\lambda(F_{total})}{M_{total} + F_{total}}M_{total}F_{total} \nonumber \\
\leq & \displaystyle -\min\{\mu_m, \mu_{mw}\} M_{total} + \lambda(F_{total})F_{total} \nonumber \\
 \leq & -\min\{\mu_m, \mu_{mw}\} M_{total} + B,\label{M-diff-ineq}
 \end{align}
where $B$ is an upper bound for $\lambda(F_{total}(t)) F_{total}(t)$, which exists since $F_{total}(t)$ is bounded and therefore so is $\lambda(F_{total}(t)) F_{total}(t)$. From the differential inequality \eqref{M-diff-ineq}, we can conclude that $M_{total}$ is bounded, too.
\end{proof}

\subsection{Boundary equilibria and their stability}
It is straightforward  to see that model \eqref{eqM}-\eqref{eqFw} has only one non-trivial \Wolb free boundary equilibrium $E^*=(M^*, F^*,0,0)$,  where $F^*$ satisfies
\begin{equation}\label{Fstar}
\lambda(F^*) = \mu_f +\mu_m,
\end{equation}
and 
\begin{equation}\label{Mstar}
M^* = \frac{\mu_f F^*}{\mu_m},
\end{equation}
 under the assumptions that $\lambda(0) > \mu_f + \mu_m$, $\lambda$ is
 a decreasing non-negative function, and
 $\lambda(F_{total})\rightarrow \lambda_{min}$
 (with $\lambda_{min}$ sufficiently small) as $F_{total} \rightarrow \infty$. 

Note that there is no \Wolb infected boundary equilibrium unless
$\tau=1$, a case that we shall treat separately later. 

\begin{thm}\label{linstabMF00}
Suppose that $\lambda$ is a monotone decreasing non-negative function
such that $\lambda(F_{total})\rightarrow \lambda_{min}$ as $F_{total} \rightarrow \infty$, with
$\lambda_{min}$ sufficiently small, $\lambda(0)>\mu_f + \mu_m$ and
\begin{equation}\label{A4}
\frac{\mu_f (1-\beta)\tau}{\mu_{fw}} <1.
\end{equation}
Then, the \Wolb free boundary equilibrium $E^*=(M^*, F^*,0,0)$ of
model \eqref{eqM}-\eqref{eqFw} is locally asymptotically stable. 
\end{thm}
\begin{proof} 
Linearisation of system \eqref{eqM}-\eqref{eqFw} at the equilibrium $E^*$ yields the following partially decoupled systems. The first system below  is just the linearisation of equations \eqref{eqMw}-\eqref{eqFw} at the steady state, which we shall use to show that $(M_w(t),F_w(t)) \rightarrow (0,0)$ as $t \rightarrow \infty$. System \eqref{lin_2} is just system \eqref{eqM}-\eqref{eqF} in the case $M_w = F_w =0$.
\begin{equation}\label{lin_1}
\left\{\begin{array}{lcl}
M_w^{\prime} & = &\displaystyle -\mu_{mw} M_w + \frac{\lambda(F^*)}{M^* + F^*}  (1-\beta)\tau (1-\gamma)M^*F_w, \\[1ex]
F_w^{\prime} & = &\displaystyle -\mu_{fw} F_w + \frac{\lambda(F^*)}{M^* + F^*}  (1-\beta)\tau M^*F_w, \\[1ex]
\end{array}\right.
\end{equation}

\begin{equation}\label{lin_2}
\left\{\begin{array}{lcl}
M^{\prime} & = &\displaystyle -\mu_{m} M + \frac{\lambda(F)}{M + F}MF, \\
F^{\prime} & = &\displaystyle -\mu_{f} F + \frac{\lambda(F)}{M + F}MF. \\
\end{array}\right.
\end{equation}
From the second equation of  \eqref{lin_1}, it is clear that if 
\begin{equation}
\frac{\lambda(F^*)}{M^* + F^*} (1-\beta)\tau M^* < \mu_{fw},\label{st-cond-ineq}
\end{equation}
then $F_w(t) \rightarrow 0$ as $t \rightarrow \infty$. Then $M_w(t) \rightarrow 0$ as $t\rightarrow \infty$ follows from the first equation of  \eqref{lin_1}. Since $M^*$ and $F^*$ are given by \eqref{Fstar} and \eqref{Mstar}, inequality \eqref{st-cond-ineq} is equivalent to assumption \eqref{A4}.

It remains to prove the local stability of $(M,F)=(M^*,F^*)$ as a solution of system \eqref{lin_2}. 
The Jacobian matrix of system \eqref{lin_2} evaluated at $(M^*,F^*)$ is given by
\begin{equation*}
J(M^*,F^*) = \frac{1}{\mu_f + \mu_m} \left( \begin{array}{ll}
-\mu_f \mu_m & \mu_f^2 + \lambda^{\prime}(F^*) F^* \mu_f \\
\mu_m^2 & -\mu_f \mu_m + \lambda^{\prime}(F^*) F^* \mu_f \\
\end{array}
 \right).
\end{equation*}
The eigenvalues $\Lambda$ of $J(M^*,F^*)$ satisfy the characteristic equation
\begin{equation*}
\Lambda^2 + (2\mu_f \mu_m - \lambda^{\prime}(F^*) F^* \mu_f)\Lambda - \lambda^{\prime}(F^*) F^* (\mu_f+\mu_m)\mu_f\mu_m=0.
\end{equation*}
Since $\lambda(\cdot)$ is a non-negative decreasing function, $\lambda^{\prime}(F^*) <0$. We have
\begin{equation*}
\Lambda_1 + \Lambda_2 = -(2\mu_f \mu_m - \lambda^{\prime}(F^*) F^* \mu_f) <0,
\end{equation*}
and 
\begin{equation*}
\Lambda_1  \Lambda_2 = - \lambda^{\prime}(F^*) F^* (\mu_f+\mu_m)\mu_f\mu_m >0,
\end{equation*}
which implies $\mbox{Re}\,\Lambda_1 <0$ and $\mbox{Re}\,\Lambda_2 <0$, so that $(M^*,F^*)$ is locally stable as a solution
of \eqref{lin_2}. Therefore, the \Wolb free equilibrium $E^*=(M^*,F^*,0,0)$ is locally asymptotically stable as a solution of
the full system \eqref{eqM}-\eqref{eqFw}.
\end{proof}

If $\tau=1$, i.e. we have complete maternal transmission of \W, then a boundary equilibrium of the form $(0,0,M_w^*,F_w^*)$ may exist. The components of such an equilibrium solution must satisfy
\begin{equation}
\begin{split}
& \mu_{mw}M_w^*=(1-\gamma)\mu_{fw}F_w^*, \\
& \mu_{fw}=\frac{\lambda(F_w^*)}{M_w^*+F_w^*}(1-\beta)M_w^*.
\end{split}
\label{120914_14}
\end{equation}
Moreover, $(1-\gamma)\mu_{fw}+\mu_{mw}=(1-\beta)(1-\gamma)\lambda(F_w^*)$ must hold. 
Next we study the linear stability of such equilibrium, showing that it is
linearly stable under condition \eqref{120914_12} below. Inequality \eqref{120914_12} does not depend on $\gamma$, the
male killing rate, but the steady state components $M_w^*$ and $F_w^*$ do depend on $\gamma$ in the manner expected (for example, $M_w^*=0$
when $\gamma=1$). Although Theorem \ref{linstab00MF} only apples if
$\tau=1$, we will be interested later on in the case when $\tau$ is just
slightly less than $1$. Then, maternal transmission is imperfect and
\Wolb infected females produce small numbers of uninfected offspring. We anticipate that as $\tau$ decreases from $1$ to a value
just less than $1$, the equilibrium $(0,0,M_w^*,F_w^*)$ shifts to another nearby position with small numbers of \Wolb uninfected individuals
and large numbers of infected ones; with no change of stability for $\tau$ close enough to $1$. The existence and stability of such an
equilibrium will be important later on when we introduce West Nile virus (WNv) disease dynamics because, at a WNv-free equilibrium with large
numbers of \Wolb infected mosquitoes, the basic reproduction number $R_0$ for WNv is likely to be less than $1$. The implication is
that \Wolb infection in mosquitoes has the potential to control WNv infection. It does so by disrupting WNv virus replication causing
WNv infected mosquitoes effectively to remain permanently (or for a very long time) in the latent stage of WNv.    
\begin{thm}\label{linstab00MF}
Assume that $\tau=1$, $\lambda$ is monotone decreasing with
$\displaystyle\lim_{F_{total}\to\infty}\lambda(F_{total})=\lambda_{min}$,
($\lambda_{min}$ sufficiently small) and 
\begin{equation*}
\lambda(0)>\frac{(1-\gamma)\mu_{fw}+\mu_{mw}}{(1-\beta)(1-\gamma)}
\end{equation*}
holds. Then, an equilibrium of the form $(M,F,M_w,F_w)=(0,0,M_w^*,F_w^*)$ exists, and it is locally stable as a solution of \eqref{eqM}-\eqref{eqFw} if 
\begin{equation}
(1-q)\mu_{fw}<(1-\beta)\mu_f. \label{120914_12}
\end{equation}
\end{thm}
\begin{proof} The proof is similar to that of Theorem~\ref{linstabMF00}. The linearisation around $(0,0,M_w^*,F_w^*)$ yields a system of linear equations for $(M,F)$ that (when $\tau=1$) are decoupled from the rest of the system. Moreover, it may be shown that $(M,F)\to (0,0)$ as $t\rightarrow\infty$ if
\begin{equation*}
\frac{(1-q)M_w^*\lambda(F_w^*)}{M_w^*+F_w^*}<\mu_f
\end{equation*}
holds, which becomes inequality \eqref{120914_12} when the equilibrium equations \eqref{120914_14} are invoked. 
Then, the $F_w$ and $M_w$ equations are considered, in the case when $\tau=1$ and $F=M\equiv 0$. Tedious computations yield that the linearisation of that system around the steady state $(M_w^*,F_w^*)$ has the Jacobian matrix equal to $((1-\gamma)\mu_{fw}+\mu_{mw})^{-1}$ times
\begin{equation*}
\left(\begin{array}{cc}
-(1-\gamma)\mu_{fw}\mu_{mw} & (1-\beta)(1-\gamma)\left[\left(\frac{1-\gamma}{1-\beta}\right)\mu_{fw}^2+\lambda'(F_w^*)F_w^*(1-\gamma)\mu_{fw}\right] \\
\mu_{mw}^2 & -\mu_{fw}\mu_{mw}+(1-\beta)\lambda'(F_w^*)F_w^*(1-\gamma)\mu_{fw}
\end{array}
\right), 
\end{equation*}
and it may be further shown that its eigenvalues both have negative real parts. Thus we conclude that the steady state $(0,0,M_w^*,F_w^*)$ is locally asymptotically stable.
\end{proof}

Next we prove that, under certain conditions, both infected and uninfected mosquitoes die out. Note, however, that $(0,0,0,0)$ is not technically an equilibrium of~\eqref{eqM}-\eqref{eqFw}.
\begin{thm}
Suppose that $\lambda$ is monotone decreasing and that $0<\lambda(F_{total})<\mu_f+\mu_m$ for all $F_{total}\geq 0$, and that 
\begin{equation}
(\mu_f+\mu_m)(1-\beta)\tau<\mu_{fw}.
\label{100914_1}
\end{equation}
 Then $(M(t),F(t),M_w(t),F_w(t))\rightarrow (0,0,0,0)$ as $t\rightarrow\infty$ if all of the four variables are sufficiently small initially. 
\end{thm}
\begin{proof} 
From \eqref{eqMw} and \eqref{eqFw},
\begin{align}
M_w'(t) & \leq  -\mu_{mw}M_w(t)+(\mu_f+\mu_m)(1-\beta)\tau(1-\gamma)F_w(t), \label{Mw-eq}\\
F_w'(t)\ & \leq -\mu_{fw}F_w(t)+(\mu_f+\mu_m)(1-\beta)\tau F_w(t).\label{Fw-eq}
\end{align}
From \eqref{Fw-eq} and \eqref{100914_1} we have $F_w(t)\rightarrow 0$ as $t\rightarrow\infty$. Then inequality \eqref{Mw-eq} implies that $M_w(t)\rightarrow 0$ also. With $M_w=F_w=0$, we are reduced to system \eqref{lin_2} and we now show that $(M(t), F(t))\rightarrow (0,0)$ as $t\rightarrow\infty$, though this result is local, i.e. for small introductions of $F$ and $M$. Note that $(M,F)=(0,0)$ is not an equilibrium of \eqref{lin_2}, due to the singularity, but we can remove that singularity by introducing the new variable $\xi=F/M$.  In terms of the variables $\xi$ and $M$, system \eqref{lin_2} becomes
\begin{equation}
\begin{split}
\xi'(t) &= -(\mu_f-\mu_m)\xi(t)+\left(\frac{1-\xi(t)}{1+\xi(t)}\right)\xi(t)\lambda(M(t)\xi(t)), \\
M'(t) &= -\mu_m M(t)+\frac{M(t)\xi(t)}{1+\xi(t)}\lambda(M(t)\xi(t)).
\end{split}
\label{100914_2}
\end{equation}
We now show that $(\xi,M)=(\xi^*,0)$ is a locally stable steady state of \eqref{100914_2}, where
\begin{equation}
\xi^*=\frac{\lambda(0)+\mu_m-\mu_f}{\lambda(0)+\mu_f-\mu_m},
\end{equation}
provided $\xi^*>0$. The latter is not automatic. However, note that, from the first equation of \eqref{lin_2}, $M'(t)\leq (\lambda(0)-\mu_m)M(t)$. Therefore, if $\lambda(0)<\mu_m$ then $M(t)\rightarrow 0$. The second of \eqref{lin_2} then gives $F'(t)\leq -\mu_f F(t)+\lambda(0)M(t)$ so that $F(t)\rightarrow 0$. By similar reasoning we arrive at the same conclusion if $\lambda(0)<\mu_f$. Therefore, we may assume henceforth that $\lambda(0)\geq \max(\mu_f,\mu_m)$ and, under these circumstances, $\xi^*>0$. The linearisation of the second equation of \eqref{100914_2} near the equilibrium $(\xi,M)=(\xi^*,0)$ reads
\begin{equation*}
M'(t)=-\mu_m M(t)+\lambda(0)\frac{\xi^*}{1+\xi^*}M(t)=\frac{1}{2}(\lambda(0)-\mu_m-\mu_f)M(t),
\end{equation*}
and therefore, since $\lambda(0)<\mu_m+\mu_f$, we have $M(t)\rightarrow 0$. In this limit the $\xi$ equation becomes
\begin{equation*}
\xi'(t) = -(\mu_f-\mu_m)\xi(t)+\left(\frac{1-\xi(t)}{1+\xi(t)}\right)\xi(t)\lambda(0)=F(\xi(t)), \\
\end{equation*}
and to show that $\xi^*$ is locally stable as a solution of this equation, it suffices to show that $F'(\xi^*)<0$. But, after some algebra,
we have
\begin{equation*}
F'(\xi^*)=-\frac{(\lambda(0))^2-(\mu_m-\mu_f)^2}{2\lambda(0)}.
\end{equation*}
To show that $F'(\xi^*)<0$ holds, it suffices to show that $\lambda(0)>|\mu_m-\mu_f|$, i.e. that
both $\lambda(0)>\mu_m-\mu_f$ and $\lambda(0)>\mu_f-\mu_m$ hold. But this follows from the fact that we are now restricting to 
the case when $\lambda(0)\geq \max(\mu_f,\mu_m)$. Therefore, the proof of the theorem is complete.
\end{proof}

\subsection{Existence of strictly positive steady states}
In this section we examine the possible existence of coexistence steady states $(M,F,M_w,F_w)$ of model \eqref{eqM}-\eqref{eqFw}, i.e. steady
states in which each component is strictly positive. It turns out that in some parameter regimes multiple coexistence steady states may exist
while, in others, there is just one or none at all. An understanding of these properties helps us to understand how one might exploit
\Wolb infection in mosquitoes to effectively control WNv. In this section we simplify by assuming that $\gamma=0$, i.e. that there is no
male killing.

At the steady state, dividing \eqref{eqM} by \eqref{eqF}, and \eqref{eqMw} by \eqref{eqFw}, we obtain
\begin{equation}
M=F\frac{\mu_f}{\mu_m},\,\, M_w=F_w\frac{\mu_{fw}}{\mu_{mw}}.\label{ss1}
\end{equation}
From \eqref{eqF} and \eqref{eqFw} we then obtain
\begin{align}
\mu_fF= & \frac{\lambda_*}{F\left(1+\frac{\mu_f}{\mu_m}\right)+F_w\left(1+\frac{\mu_{fw}}{\mu_{mw}}\right)} \nonumber \\
 & \times\left(\frac{\mu_f}{\mu_m}F^2+(1-\beta)(1-\tau)\left(\frac{\mu_f}{\mu_m}FF_w+\frac{\mu_{fw}}{\mu_{mw}}F_w^2\right)+(1-q)\frac{\mu_{fw}}{\mu_{mw}}FF_w\right), \label{ss2} \\
\mu_{fw}F_w= &\frac{\lambda_*}{F\left(1+\frac{\mu_f}{\mu_m}\right)+F_w\left(1+\frac{\mu_{fw}}{\mu_{mw}}\right)}\left((1-\beta)\tau\left(\frac{\mu_f}{\mu_m}FF_w+\frac{\mu_{fw}}{\mu_{mw}}F_w^2\right)\right), \label{ss3}
\end{align}
respectively, where $\lambda_*=\lambda(F+F_w)$. From \eqref{ss3} we have
\begin{equation}
F_w\left(\mu_{fw}+\frac{\mu^2_{fw}}{\mu_{mw}}\right)+F\left(\mu_{fw}+\frac{\mu_{fw}\mu_f}{\mu_m}\right)=\lambda_*
\left(F(1-\beta)\tau\frac{\mu_f}{\mu_m}+F_w(1-\beta)\tau\frac{\mu_{fw}}{\mu_{mw}}\right).\label{ss3-2}
\end{equation}
From \eqref{ss3-2} we obtain
\begin{equation}
F_w\kappa_1(\lambda_*)=F\kappa_2(\lambda_*),\label{ss4}
\end{equation}
where 
\begin{equation}\label{ss5}
\kappa_1(\lambda_*)=\mu_{fw}+\frac{\mu_{fw}^2}{\mu_{mw}}-\lambda_*(1-\beta)\tau\frac{\mu_{fw}}{\mu_{mw}},\quad 
\kappa_2(\lambda_*)=-\mu_{fw}-\mu_{fw}\frac{\mu_f}{\mu_m}+\lambda_*(1-\beta)\tau\frac{\mu_{f}}{\mu_{m}}.
\end{equation}
Note that if at the strictly positive steady state, we have $\kappa_1(\lambda_*)=0$, then this necessarily implies that
$\kappa_2(\lambda_*)=0$. This is only possible if 
\begin{equation}
\frac{\mu_{fw}}{\mu_f}=\frac{\mu_{mw}}{\mu_m}\label{ss5-2}
\end{equation}
holds. This case is excluded from Theorem~\ref{nomulteqtau1} but is treated in the next subsection.

If \eqref{ss5-2} does not hold then, using~\eqref{ss5}, from~\eqref{ss2}  we obtain
\begin{align}
0= & \kappa^2_1(\lambda_*)\left(\mu_f+\frac{\mu^2_f}{\mu_m}-\lambda_*\frac{\mu_f}{\mu_m}\right)-\kappa^2_2(\lambda_*)\lambda_*(1-\beta)(1-\tau)\frac{\mu_{fw}}{\mu_{mw}}  \nonumber  \\ 
& +\kappa_1(\lambda_*)\kappa_2(\lambda_*)\left(\mu_f+\frac{\mu_f\mu_{fw}}{\mu_{mw}}-\lambda_*(1-\beta)(1-\tau)\frac{\mu_f}{\mu_m}-\lambda_*(1-q)\frac{\mu_{fw}}{\mu_{mw}}\right).\label{ss6}
\end{align}
The right hand side of~\eqref{ss6} is, in general, a cubic
polynomial in $\lambda_*$. If there exists a positive root
$\lambda^1_*$, then  
since $\lambda$ is a strictly monotone function, a corresponding unique $F^1+F^1_w$ value may be found. From~\eqref{ss2} 
we may then determine a unique solution $(F^1,F^1_w)$. Further
analytic progress is possible in certain particular cases of interest, which we
now investigate. The first concerns the case when $\tau=1$, or when $\tau$ is very close to $1$, meaning that maternal transmission of \Wolb is complete or nearly complete. This is in fact the biologically relevant case for a number of CI inducing \Wolb strains in mosquito species treated in the literature, see e.g. \cite{Engelstadter2004}. This leads us to expect the existence of a steady state
of \eqref{eqM}-\eqref{eqFw} with large numbers of \Wolb infected mosquitoes and few, or no, uninfected ones. The stability of such a
steady state still depends on the other parameter values, and it will be stable if there is a high probability of mating between infected
males and uninfected females resulting in no offspring (the effect of CI), i.e. $q$ is
close to $1$; see also inequality \eqref{120914_12} for the case $\tau=1$.
The existence of a stable steady state of \eqref{eqM}-\eqref{eqFw} with the above mentioned properties is important because the low
number of \Wolb uninfected females implies that the quantity $F_s^*$, featuring in the first term of the parameter $R_0$ defined later
in \eqref{110914_1}, is small. The likelihood of $R_0$ being less than $1$ (the condition for WNv-eradication) depends mostly
on that first term involving $F_s^*$, since the second term in \eqref{110914_1} involves a small parameter $\varepsilon$ and is
automatically small. For these reasons, we are interested in stable steady states of model \eqref{eqM}-\eqref{eqFw} of the form
$(M^*,F^*,M_w^*,F_w^*)$, with $M^*$ and $F^*$ small compared to $M_w^*$ and $F_w^*$ (and, ideally, $M^*=F^*=0$). 
Therefore, referring to Theorems \ref{linstabMF00} and \ref{linstab00MF}, and restricting for now to the case $\tau=1$, we ideally would like the boundary equilibrium $(M^*,F^*,0,0)$ of Theorem \ref{linstabMF00} to be unstable, and the boundary equilibrium $(0,0,M_w^*,F_w^*)$ of Theorem~\ref{linstab00MF} to be linearly stable. The conditions for this, when $\tau=1$, are that the inequalities $\mu_f(1-\beta)>\mu_{fw}$ and
$(1-q)\mu_{fw}<(1-\beta)\mu_f$ should hold simultaneously, but note that the second of these follows from the first. For $\tau=1$, guided by elementary competition theory, instability of one boundary equilibrium and
stability of the other suggests that there will be no coexistence equilibrium, and this is what we prove in Theorem \ref{nomulteqtau1} below. If $\tau$ is decreased from $1$ to a value slightly below $1$, there is no longer a boundary equilibrium with only \Wolb infected mosquitoes present.
What happens is that the equilibrium $(0,0,M_w^*,F_w^*)$, which exists when $\tau=1$, moves to another nearby point in $\mathbb{R}^4_+$, so
that \Wolb infected mosquitoes now coexist with uninfected ones, the former being dominant. Very importantly, this will be the only coexistence
steady state if $\tau$ is sufficiently close to $1$, and it is a desirable steady state for WNv eradication because particular \Wolb strains can significantly reduce WNv virus replication in mosquitoes, see \cite{Hussain2013}. 

\begin{thm}\label{nomulteqtau1}
Suppose that $\gamma=0$, $\tau=1$, $\mu_f(1-\beta)>\mu_{fw}$,
$\mu_m\mu_{fw}\neq\mu_f\mu_{mw}$, and that $\lambda$ is a strictly positive decreasing function with
$\lambda(\infty)=\lambda_{min}$ (and $\lambda_{min}$ sufficiently small). Then system
\eqref{eqM}-\eqref{eqFw} has no coexistence equilibrium with
$M^*,F^*,M_w^*,F_w^*>0$.  

If the foregoing hypotheses hold, except that $\tau$ is slightly less than $1$, then system \eqref{eqM}-\eqref{eqFw} has precisely one coexistence equilibrium in which \Wolb uninfected mosquitoes exist in very small numbers relative to \Wolb infected ones.
\end{thm}
\begin{proof} Since we assume $\tau=1$, the form of \eqref{ss6} simplifies and in fact we may cancel $\kappa_1(\lambda_*)$ since we
seek equilibria in which $M^*,F^*,M_w^*,F_w^*>0$. After some further algebra, we find that there is just one value for $\lambda_*$, given
by 
\begin{equation}
\begin{split}
\lambda_*\,q(1-\beta)\frac{\mu_{fw}\mu_f}{\mu_{mw}\mu_m} = &
(1-\beta)\left(\frac{\mu_f\mu_{fw}}{\mu_{mw}}-\frac{\mu_f^2}{\mu_m}\right) 
+q\frac{\mu_f\mu_{fw}^2}{\mu_m\mu_{mw}}
\\
 & +\frac{\mu_f\mu_{fw}}{\mu_m}-(1-q)\frac{\mu_{fw}^2}{\mu_{mw}}.
\end{split}
\label{081014_1}
\end{equation}
If $\lambda_*\leq 0$ then the equation $\lambda_*=\lambda(F+F_w)$ cannot
be solved for $F+F_w$, so we may restrict to the case that
$\lambda_*>0$. Recalling that $\kappa_1$ and $\kappa_2$ are defined
by \eqref{ss5}, we find, with $\lambda_*$ given by \eqref{081014_1} and with $\tau=1$, that
\begin{equation}
\kappa_1(\lambda_*)=\left(1-\frac{\mu_m\mu_{fw}}{\mu_f\mu_{mw}}\right)
\left[\mu_{fw}+\frac{1}{q}\big((1-\beta)\mu_f-\mu_{fw}\big)\right],
\label{081014_2}
\end{equation}
and
\begin{equation}
\kappa_2(\lambda_*)=\frac{1}{q}\big((1-\beta)\mu_f-\mu_{fw}\big)\left(
1-\frac{\mu_f\mu_{mw}}{\mu_m\mu_{fw}}\right).
\label{081014_3}
\end{equation}
Since $\mu_f(1-\beta)>\mu_{fw}$ it follows that the sign of the product $\kappa_1(\lambda_*)\kappa_2(\lambda_*)$ is the same as the sign of
\begin{equation*}
\left(1-\frac{\mu_m\mu_{fw}}{\mu_f\mu_{mw}}\right)
\left(1-\frac{\mu_f\mu_{mw}}{\mu_m\mu_{fw}}\right)
\end{equation*}
and, since we assume $\mu_m\mu_{fw}\neq\mu_f\mu_{mw}$, it follows that $\kappa_1(\lambda_*)\kappa_2(\lambda_*)<0$. This makes it impossible
to find $F>0$ and $F_w>0$ satisfying~(\ref{ss4}), and so there is no coexistence equilibrium. 

Next we prove the second assertion of the theorem. Let
$\tau=1-\varepsilon$. Since we expect the equilibrium
  $(0,0,M_w^*,F_w^*)$, which exists when $\tau=1$, to move to another
  nearby point when $\tau=1-\epsilon$, for $\epsilon$ sufficiently
  small, we seek an equilibrium of \eqref{eqM}--\eqref{eqFw} of the form
\begin{align*}
M(\varepsilon)= & \varepsilon M^{(1)}+\varepsilon^2 M^{(2)}+\cdots, \\
F(\varepsilon)= & \varepsilon F^{(1)}+\varepsilon^2 F^{(2)}+\cdots, \\
M_w(\varepsilon)= & M_w^*+\varepsilon M_w^{(1)}+\varepsilon^2 M_w^{(2)}+\cdots, \\
F_w(\varepsilon)= & F_w^*+\varepsilon F_w^{(1)}+\varepsilon^2 F_w^{(2)}+\cdots.
\end{align*}
Coefficients of $\varepsilon$ yield
\begin{equation}
\mu_m M^{(1)}=\frac{\lambda(F_w^*)}{M_w^*+F_w^*}\left[(1-\beta)M_w^*F_w^*+(1-q)M_w^*F^{(1)}\right], \label{perturb1}
\end{equation}
and
\begin{equation}
\mu_f F^{(1)}=\frac{\lambda(F_w^*)}{M_w^*+F_w^*}\left[(1-\beta)M_w^*F_w^*+(1-q)M_w^*F^{(1)}\right]. \label{perturb2}
\end{equation}
Using
\begin{equation*}
\mu_{fw}=\frac{\lambda(F_w^*)}{M_w^*+F_w^*}(1-\beta)M_w^*,
\end{equation*}
\eqref{perturb2} reads
\begin{equation*}
\mu_fF^{(1)}=\mu_{fw}F_w^*+\frac{\mu_{fw}(1-q)}{1-\beta}F^{(1)},
\end{equation*}
so that 
\begin{equation}
F^{(1)}=\frac{(1-\beta)\mu_{fw}F_w^*}{(1-\beta)\mu_f-(1-q)\mu_{fw}}.\label{perturb3}
\end{equation}
Note that $F^{(1)}>0$ because of the assumption $\mu_f(1-\beta)>\mu_{fw}$, hence from \eqref{perturb1} we conclude that $M^{(1)}>0$ holds, and 
\begin{equation}
M^{(1)}=\frac{\mu_{fw}\mu_f(1-\beta)F_w^*}{\mu_m\left[(1-\beta)\mu_f-(1-q)\mu_{fw}\right]}.\label{perturb4}
\end{equation} 
In conclusion, if $\tau$ is reduced from $1$ to the value $1-\varepsilon$ then the equilibrium $(0,0,M_w^*,F_w^*)$ moves to 
$(\varepsilon M^{(1)},\varepsilon F^{(1)},M_w^*+\varepsilon M_w^{(1)},F_w^*+\varepsilon F_w^{(1)})$, with $M^{(1)}$ and $F^{(1)}$ given by 
\eqref{perturb4} and \eqref{perturb3}, respectively.
\end{proof}
Note that from the proof of Theorem \ref{nomulteqtau1} we can see that at the equilibrium the male/female ratio for \Wolb uninfected mosquitoes 
is given approximately by $\frac{\mu_f}{\mu_m}$.

Theorem~\ref{nomulteqtau1} excludes the case when
$\mu_m\mu_{fw}=\mu_f\mu_{mw}$, which we treat now. In particular, we
assume that $\gamma=0$, and $\frac{\mu_{f}}{\mu_m}=\mu=\frac{\mu_{fw}}{\mu_{mw}}$.  
In this situation, we have $M=\mu F$, and $M_w=\mu F_w$.
From equations \eqref{eqF} and \eqref{eqFw} we obtain
\begin{align}
\mu_fF= & \frac{\mu\lambda_*}{(1+\mu)(F+F_w)}\left(F^2+(1-\beta)(1-\tau)\left(FF_w+F^2_w\right)+(1-q)FF_w\right), \label{ss7} \\
\mu_{fw}F_w= & \frac{\mu\lambda_*}{(1+\mu)(F+F_w)}\left((1-\beta)\tau\left(FF_w+F^2_w\right)\right). \label{ss8}
\end{align}
From equation \eqref{ss8} we find that
\begin{equation}
\lambda_*=\lambda(F+F_w)=\frac{(1+\mu)\mu_{fw}}{\mu(1-\beta)\tau},\label{ss9}
\end{equation}
hence for the existence of a coexistence steady state it is necessary that $\lambda(0)>\frac{(1+\mu)\mu_{fw}}{\mu(1-\beta)\tau}$, in which case 
there exists a unique $c=F+F_w$, such that $\lambda(c)=\frac{(1+\mu)\mu_{fw}}{\mu(1-\beta)\tau}$ holds.
Then, from equation \eqref{ss7} we obtain
\begin{equation}
\frac{\mu_f}{\mu_{fw}}c(1-\beta)\tau F=F^2+(1-q)FF_w+c(1-\beta)(1-\tau)F_w, \label{ss10}
\end{equation}
from which, using $F_w=c-F$, we obtain
\begin{equation}\label{ss11}
0=F^2q+Fc\left((1-q)-(1-\beta)(1-\tau)-(1-\beta)\tau\frac{\mu_f}{\mu_{fw}}\right)+c^2(1-\beta)(1-\tau).
\end{equation}
For the existence of the positive steady state we need to guarantee that the quadratic equation above has (at least one) positive (real) solution, 
and that the solution is less than $c$. 

From equation \eqref{ss11} it is clear that in case of complete maternal transmission, i.e. for $\tau=1$ there cannot be more than one 
coexistence steady state. In this case, from equation \eqref{ss11} we obtain
\begin{equation}
F=\frac{c\left((1-\beta)\frac{\mu_f}{\mu_{fw}}-(1-q)\right)}{q}.\label{ss12}
\end{equation}
Therefore, if
\begin{equation}
1-q<(1-\beta)\frac{\mu_f}{\mu_{fw}}<1\label{ss13}
\end{equation}
holds, then we have $0<F<c$, and if $\lambda(0)>\frac{(1+\mu)\mu_{fw}}{\mu(1-\beta)\tau}$ also holds, then a unique strictly positive steady state exists.

Circumstances under which \eqref{ss13} is likely to hold include that $q$ is sufficiently close to $1$ and, at the same time, $\beta$ is sufficiently close to $1$ or $\frac{\mu_f}{\mu_{fw}}$ is less than $1$. The biological interpretation is clear: for the existence of a coexistence steady state, the fertility cost (or mortality increase) due to \Wolb infection should be sufficiently large.

It is clear from \eqref{ss11} that we can never have more than two coexistence steady states. 
On the other hand it is interesting to show that for some realistic parameter values it is possible to have two coexistence steady states. 

To this end we consider the case $q=1$, $\frac{\mu_f}{\mu_{fw}}=1$,
and we assume that $\tau\ne 1$, $\beta\ne 1$.  
In this case, from \eqref{ss11}, we have
\begin{equation}
F_{1/2}=c\frac{1-\beta}{2}\left(1\pm \sqrt{1-\frac{4(1-\tau)}{1-\beta}}\right).
\end{equation} 
That is, for $0<F_{1/2}<c$ to hold, we need to assume that 
\begin{equation}
4(1-\tau)<1-\beta,\quad 1+\sqrt{1-\frac{4(1-\tau)}{1-\beta}}<\frac{2}{1-\beta}
\end{equation}
hold simultaneously. It is easy to verify that this can be achieved for any $c$, for example with $\tau=0.99$, $\beta=0.5$.
Note that the condition $\frac{\mu_f}{\mu_{fw}}=1$ can be relaxed,
too. Also, by continuity arguments, it follows that for parameter
values close enough we still have two coexistence steady states.
We summarise our findings in the following proposition.
\begin{proposition}
In the case when $\gamma=0$, $q=1$ and
$\frac{\mu_f}{\mu_m}=\frac{\mu_{fw}}{\mu_{mw}}$, there exists a set of
values for the remaining parameters such that system
\eqref{eqM}-\eqref{eqFw} admits two coexistence steady states. 
\end{proposition}
In summary, we have shown that our {\it Wolbachia}
  model \eqref{eqM}-\eqref{eqFw} may exhibit all of the three
  qualitatively different possible scenarios, i.e. when there are
  $0,1$ or $2$ coexistence steady states.

\section{Model incorporating West Nile virus (WNv)}
There has been considerable recent interest in West Nile virus (WNv),
with the great majority of mathematical papers on the topic having
appeared in the last 15 years. Numerous types of models have appeared,
some including spatial effects and others giving consideration to
issues such as age-structure in hosts, optimal control or backward bifurcation. See, for example,
Blayneh et~al~\cite{Blayneh2010}, Bowman et~al~\cite{Bowman2005},
Gourley et~al~\cite{Gourley2007}, Lewis et~al~\cite{Lewis2006} and
Wonham and Lewis~\cite{Wonham2008}.

We introduce the following model as an extension of model
\eqref{eqM}-\eqref{eqFw} to include WNv dynamics. Our
  model for WNv dynamics has similarities to that in
  Bergsman~et~al~\cite{Bergsman2016} with one resident bird 
  population. Hence, as in \cite{Bergsman2016} and in some references therein,
  we compartmentalise the vector and bird population into SEI and SEIR classes,
  respectively. Taking into account our earlier model the complete
  WNv-{\it Wolbachia} mosquito-bird population model takes the
  following form.
\begin{align}
F_s^{\prime}  &=  \displaystyle \frac{\lambda(F_{total})}{M+ M_w + F + F_w} (MF + (1-\beta)(1-\tau)(MF_w + M_wF_w) + (1-q)M_wF ) \nonumber \\
&\quad - \mu_f F_s - \alpha_f p_{bf} F_s \frac{B_i}{B_{total}}, \nonumber \\
F_e^{\prime} & = \displaystyle  \alpha_f p_{bf} F_s \frac{B_i}{B_{total}}  -\mu_f F_e - \nu_f F_e, \nonumber \\
F_i^{\prime} &=  \displaystyle \nu_f F_e - \mu_f F_i,\nonumber \\
F_{ws}^{\prime} & = \displaystyle \frac{\lambda(F_{total})}{M+ M_w + F + F_w} (1-\beta)\tau(MF_w + M_wF_w)  - \mu_{fw} F_{ws} - \alpha_{fw} p_{bf} F_{ws} \frac{B_i}{B_{total}},\nonumber \\
F_{we}^{\prime} & = \displaystyle  \alpha_{fw} p_{bf} F_{ws} \frac{B_i}{B_{total}}  -\mu_{fw} F_{we} - \varepsilon \nu_f F_{we}, \nonumber \\
F_{wi}^{\prime}& =  \displaystyle \varepsilon \nu_{f} F_{we} - \mu_{fw} F_{wi},\nonumber \\
B_s^{\prime}& =  \displaystyle \Pi(B_{total} ) - \mu_b B_s - \left(\alpha_f p_{fb}F_i \frac{B_s}{B_{total}} + \alpha_{fw} p_{fb} F_{wi} \frac{B_s}{B_{total}} \right),\nonumber \\
B_e^{\prime} &=  \displaystyle  \alpha_f p_{fb}F_i \frac{B_s}{B_{total}} + \alpha_{fw} p_{fb} F_{wi} \frac{B_s}{B_{total}} - \mu_b B_e - \nu_b B_e, \nonumber \\
B_i^{\prime} &=  \displaystyle \nu_b B_e - \mu_b B_i - \mu_{bi}B_i - \nu_i B_i,\nonumber \\
B_r^{\prime} &=  \displaystyle \nu_iB_i - \mu_b B_r,\nonumber \\
M^{\prime}  & = \displaystyle \frac{\lambda(F_{total})}{M+ M_w + F + F_w} (MF + (1-\beta)(1-\tau)(MF_w + M_wF_w) + (1-q)M_wF ) - \mu_m M,\nonumber \\
M_w^{\prime} &= \displaystyle \frac{\lambda(F_{total})}{M+ M_w + F + F_w} (1-\beta)\tau(1-\gamma)(MF_w + M_wF_w)  - \mu_{mw} M_w,\label{F}
\end{align}
where
\begin{align*}
B_{total}& = B_s +B_e +B_i +B_r,  \\
F_{total} & = F+ F_w,\quad F_w = F_{ws} + F_{we} + F_{wi},\quad F = F_s + F_e + F_i.
\end{align*}
Some of the parameters of system \eqref{F} are defined after \eqref{eqM}-\eqref{eqFw}; the rest are defined as follows:
\begin{itemize}
\item $\alpha_f$: biting rate of female \Wolb uninfected mosquitoes;
\item $\alpha_{fw}$: biting rate of female \Wolb infected mosquitoes;
\item $p_{bf}$: transmission probability of WNv from infectious birds to WNv-susceptible female mosquitoes;
\item $p_{fb}$: transmission probability of WNv from WNv-infectious female mosquitoes to susceptible birds;
\item $\nu_f$: per-capita transition rate of WNv-exposed female
  \Wolb uninfected mosquitoes to the infectious stage of WNv;
\item $\varepsilon \in [0,1]$: small parameter modelling increased time
  that \Wolb infected mosquitoes spend in the latent stage of WNv,
  due to the tendency of \Wolb infection to hamper the replication
  of WNv in mosquitoes;
\item $\nu_b$: per-capita transition rate of WNv-exposed birds to the
  infectious stage of WNv;
\item $\nu_i$: per-capita rate at which infectious birds recover;
\item $\mu_b$: per-capita natural death  rate for birds;
\item $\mu_{bi}$: per-capita WNv-induced death rate for infectious birds.
\end{itemize}
In this section, it must be emphasized that we are considering two different kinds
of infection. Mosquitoes may be infected by either \Wolb or WNv,
or both. WNv infection is assumed possible only for female mosquitoes
(since it is females that bite) and is modelled using an SEI
(susceptible-exposed-infectious) approach with subscripts $s$, $e$ and
$i$. The variables $F_s$, $F_e$ and $F_i$ denote the numbers of
\Wolb uninfected mosquitoes that have susceptible, exposed and
infectious status with respect to WNv. A subscript $w$ indicates
\Wolb infection, so that $F_{ws}$, $F_{we}$ and $F_{wi}$ denote the numbers of
\Wolb infected mosquitoes that have susceptible, exposed and
infectious status with respect to WNv. The variables $M$ and $M_w$
are the numbers of \Wolb uninfected and \Wolb infected male
mosquitoes, none of which have WNv. Birds are only susceptible to WNv and their numbers are
given by the variables $B_s$, $B_e$, $B_i$ and $B_r$ denoting susceptible,
exposed, infectious and recovered birds. Many of the terms of
system \eqref{F} are also present in \eqref{eqM}-\eqref{eqFw} without
change, here we just discuss the extra terms that model
the addition of WNv dynamics. WNv-susceptible mosquitoes, whether
\Wolb infected or not, acquire WNv infection by biting infectious
birds; this is modelled via the last term in the first and fourth
equations of \eqref{F} using the idea of mass action normalised by
total host density, the biting rates (the $\alpha$ parameters defined
above) having been separated out, rather than being absorbed into the
transmission coefficients $p_{bf}$ and $p_{fb}$ as is often
customary. Having acquired WNv infection from a bird, a mosquito
enters the latent phase of WNv and is classed as an exposed
mosquito. Exposed mosquitoes become WNv-infectious at rates $\nu_f
F_e$ and $\varepsilon\nu_f F_{we}$ for \Wolb uninfected and
\Wolb infected mosquitoes, respectively. In the latter, the presence
of $\varepsilon\in[0,1]$ models the tendency of \Wolb infected
mosquitoes that have contracted WNv infection to spend a greater
amount of time in the latent stage of WNv, since \Wolb infection
tends to block WNv replication making it less likely that such a
mosquito would ever become WNv-infectious. Of course, we are at
liberty to take $\varepsilon$ very small indeed, with the implication
that the \Wolb infected mosquito spends so long in the latent
stage of WNv that it probably dies in that stage. This is our approach
to modelling the blocking of WNv replication by \W.

Birds acquire WNv from bites by WNv-infectious mosquitoes, which may
or may not have \Wolb infection as well. Thus there are two
infection rates for birds, these can be found in the right hand side
of the seventh equation of \eqref{F}, and also in the eighth equation since
birds initially enter the exposed stage of WNv. This has a mean
duration of $1/\nu_b$ for birds, after which they become
WNv-infectious. Birds may recover from WNv, at a per-capita rate
$\nu_i$. Note that, for birds, death due to WNv is modelled using a
separate parameter $\mu_{bi}$ to distinguish from natural death,
accounted for by $\mu_b$.
The function $\Pi(B_{total})$ is the birth rate function for birds.

The approach we use here to model the latency stage of WNv (in either
birds or mosquitoes) is not the
only possible approach. Our approach permits individuals to spend
different amounts of time in the latency stage, and we may only speak
of the mean time spent in that stage. There are other approaches in
which all individuals of a particular status (for example, all
\Wolb uninfected mosquitoes) spend the same amount of time in the
latent stage of WNv. The time could be different for
\Wolb infected mosquitoes. These approaches result in models with
time delays.

\subsection{Local stability of the WNv-free equilibria}
Equilibria of system \eqref{F} may exist in which WNv is absent. Such WNv-free equilibria
include the equilibrium $(M^*,F^*,0,0)$ considered in
Theorem \ref{linstabMF00}, in which both WNv and \Wolb are absent,
and equilibria in which WNv is absent but \Wolb are present. We
show that multiple WNv-free equilibria may coexist that have both
\Wolb uninfected and \Wolb infected mosquitoes, we present a
necessary and sufficient condition for any particular WNv-free
equilibrium to be locally stable, and we show that the most likely
scenario for eradication of WNv is to have large numbers of
\Wolb infected mosquitoes, with solutions of system \eqref{F}
evolving to a WNv-free equilibrium that has large numbers of
\Wolb infected mosquitoes and relatively few uninfected
ones. Theorem \ref{linstabWNvfreeeq} applies to any WNv-free
equilibrium, of which there may be several. Of course, we may have
$R_0<1$ at one WNv-free equilibrium and $R_0>1$ at another. It depends
on the values of $F_s^*$, $F_{ws}^*$ and $B_s^*$ for the particular
WNv-free equilibrium under consideration. For clarity of exposition,
we include as a hypothesis that the equilibrium be stable to the
subset of perturbations in which WNv is absent (i.e. stable as
a solution of the subsystem \eqref{eqM}-\eqref{eqFw}), rather than
including explicit conditions for stability of an equilibrium as a
solution of that subsystem.  The latter stability problem is a tedious
one in its own right and is under consideration elsewhere in this
paper. Theorem \ref{linstabWNvfreeeq}, in the form presented below,
highlights clearly the particular role played by $R_0$. 
\begin{thm}\label{linstabWNvfreeeq}
Let $F_s^*$, $F_{ws}^*$ and $B_s^*$ be the equilibrium values for the female susceptible (to WNv) \Wolb uninfected and \Wolb infected  mosquitoes and susceptible birds, in any WNv-free equilibrium. Let
\begin{equation}
R_0=\frac{\nu_b\nu_fp_{fb}p_{bf}}{(\mu_b+\mu_{bi}+\nu_i)(\mu_b+\nu_b)}\left(\frac{\alpha_f^2(F_s^*/B_s^*)}{\mu_f(\mu_f+\nu_f)}+
\frac{\varepsilon \alpha_{fw}^2(F_{ws}^*/B_s^*)}{\mu_{fw}(\mu_{fw}+\epsilon \nu_f)}\right).
\label{110914_1}
\end{equation}
Then, if $R_0<1$, the WNv-free equilibrium under consideration is locally stable as a solution of the full system \eqref{F}, if it is stable to perturbations in which the exposed and infectious variables remain zero.
\end{thm}
\begin{proof} At any WNv-free equilibrium, the linearisation of system \eqref{F} decouples to some extent making it sufficient to show that, when $R_0<1$, each component of the solution of the following system:
\begin{eqnarray}
F_e' & = & \frac{\alpha_fp_{bf}F_s^*}{B_s^*}B_i-(\mu_f+\nu_f)F_e, \nonumber \\
F_i' & = & \nu_f F_e - \mu_f F_i, \nonumber \\
F_{we}' & = & \frac{\alpha_{fw}p_{bf}F_{ws}^*}{B_s^*} B_i-(\mu_{fw}+\varepsilon \nu_f)F_{we}, \nonumber \\
F_{wi}' & = & \varepsilon \nu_f F_{we}-\mu_{fw}F_{wi}, \label{110914_4} \\
B_e' & = & \alpha_f p_{fb}F_i+\alpha_{fw}p_{fb}F_{wi}-(\mu_b+\nu_b)B_e, \nonumber \\
B_i' & = & \nu_b B_e - (\mu_b+\mu_{bi}+\nu_i)B_i, \nonumber
\end{eqnarray}
tends to zero. It is taken as a hypothesis that the susceptible variables then approach their respective steady state values. Note that system \eqref{110914_4} has a structure that allows the application of Theorem 5.5.1 in Smith \cite{Smith95}, making it possible to restrict attention to the real roots of the characteristic equation associated with \eqref{110914_4}. That  characteristic equation, corresponding to trial solutions with temporal dependence $\exp(\lambda t)$, is most easily analysed when written in the form
\begin{equation}
\begin{split}
&(\lambda+\mu_b+\mu_{bi}+\nu_i)(\lambda+\mu_b+\nu_b)B_s^* \\
& = \nu_b\nu_fp_{fb}p_{bf}\left[\frac{\alpha_f^2F_s^*}{(\lambda+\mu_f)(\lambda+\mu_f+\nu_f)}+\frac{\varepsilon\alpha_{fw}^2F_{ws}^*}{(\lambda+\mu_{fw})(\lambda+\mu_{fw}+\epsilon \nu_f)}\right].
\end{split}
\label{110914_5}
\end{equation}
As functions of the real variable $\lambda$, the right hand side of \eqref{110914_5} is decreasing, at least for $\lambda\geq 0$, while the left hand side is a quadratic with two real negative roots. A simple graphical argument shows that if the left  hand side exceeds the right hand side when $\lambda=0$ (i.e., if $R_0<1$, with $R_0$ defined by \eqref{110914_1}), then any real roots of the characteristic equation are negative which, since only the real roots need to be considered, implies that each component of the solution of \eqref{110914_4}  approaches zero as $t\rightarrow\infty$. The proof of the theorem is now complete. 
\end{proof}

\section{Numerical simulations}
In the simulations we set $\lambda(F_{total}):=re^{-F_{total}/k}$,
where $r$ is the maximum per-capita mosquito egg-laying rate, and $k$
measures intra-specific competition among female
mosquitoes. It should be noted that our model assumes that the
mosquito population persists annually, as for example in the tropical climates of South-East Asia.

\begin{table}\begin{tabular}{lll}
\hline
Symbol & Definition & Value \\
\hline
$\mu_m$ & Per-capita mortality rate of male mosquitoes & 1/20  \\
$\mu_f$ & Per-capita mortality rate of female mosquitoes & 1/20   \\
$\mu_{wm}$ & Per-capita mortality rate of {\it W}-infected male mosquitoes & 1/20   \\
$\mu_{wf}$ & Per capita mortality rate of {\it W}-infected female mosquitoes & 1/20  \\
$r$ & Maximum per-capita mosquito egg-laying rate & 30  \\
$k$ & Competition coefficient for mosquitoes & 5000   \\
$\beta$ & Fitness cost of {\it W}-infection on reproduction &  $\in [0,1]$ \\
$\tau$ & Maternal transmission rate of \Wolb & $\in [0,1]$   \\
$q$ & Strength of CI due to {\it W}-infection & $\in [0,1]$   \\
$\gamma$ & Male killing rate due to {\it W}-infection & $\in [0,1]$  \\
$\alpha_f $ & Per-capita {\it W}-free mosquito biting rate  & 0.09  \\
$\alpha_{fw} $ & Per-capita {\it W}-infected mosquito biting rate & 0.09  \\
$p_{bf}$ & WNv transmission coefficient from birds to mosquitoes & 0.16 \\
$p_{fb} $ & WNv transmission coefficient from mosquitoes to birds & 0.88 \\
$\mu_b$ & Natural per-capita mortality rate of birds & $1/(365\times 3)$  \\
$\mu_{bi}$ & WNv-induced per-capita death rate of birds & 0.1   \\
$\Pi(B)$ & Birth rate of birds & 100/365  \\
$\nu_f$ & Per-capita rate at which {\it W}-free mosquitoes &  \\
&  complete WNv-latency and become WNv-infectious  &  $\in [0,1]$ \\
$\varepsilon$ & $\varepsilon\, \nu_f$ is the per-capita rate at which {\it W}-infected mosquitoes &   \\
 &  complete WNv-latency and become WNv-infectious  &       \\
$\nu_b$ & Per-capita rate at which exposed birds become infectious  & 0.2 \\
$\nu_i$ & Per capita rate at which infectious birds recover   & 0.2  \\
\hline
\end{tabular}
\vspace{3mm}
\caption{Definition of parameters. Time is measured in days, and rates
  in day$^{-1}$. $W$ stands for \Wolb so that, for example,
  $W$-infection means \Wolb infection. Parameter values have been
  chosen purely to demonstrate possible solution behaviour and are
  not based on data.}\label{table1} 
\end{table}

\begin{figure}[H]
\centering
\subfigure{\includegraphics[width=6.6cm]{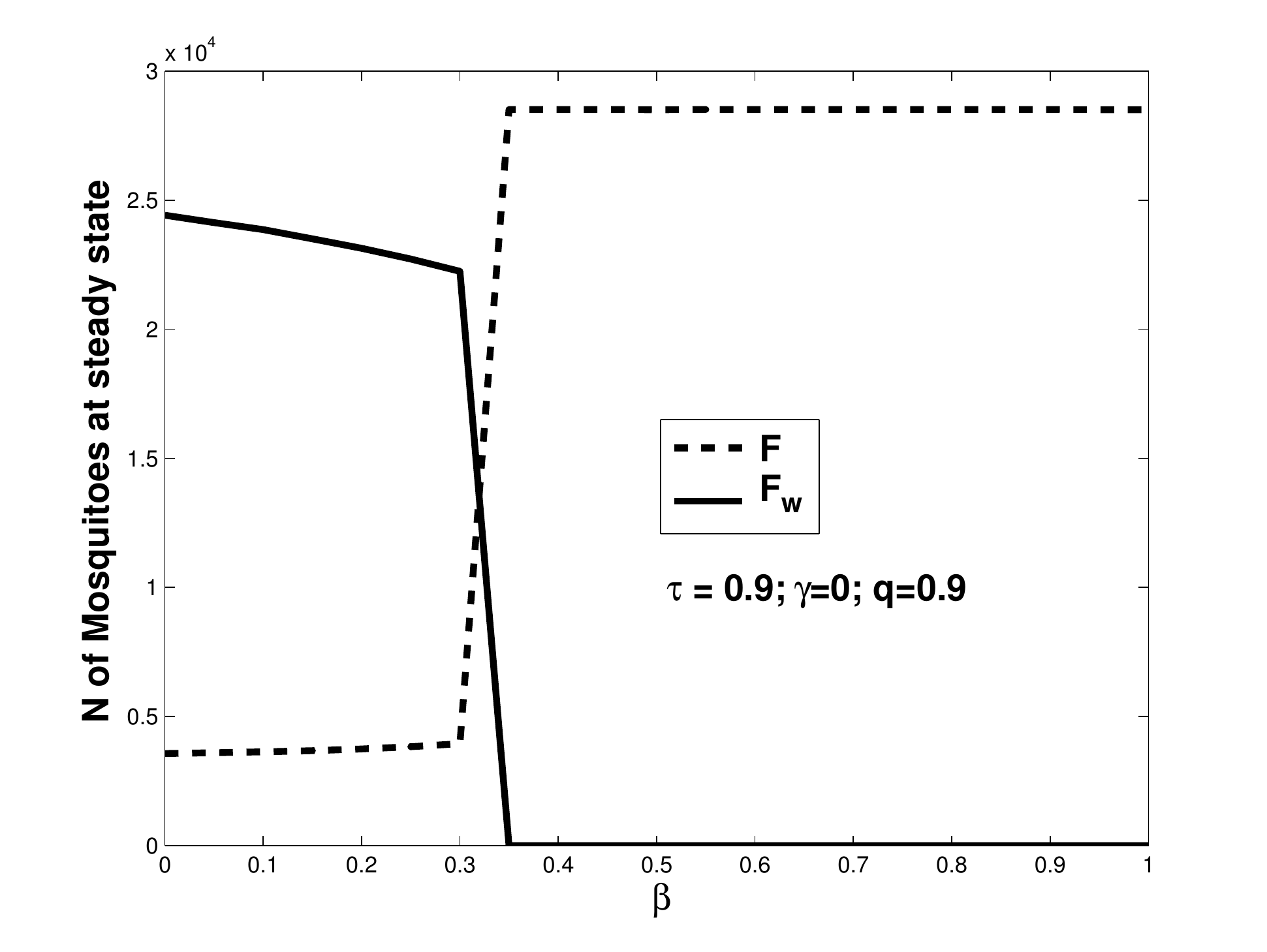}}
\subfigure{\includegraphics[width=6.6cm]{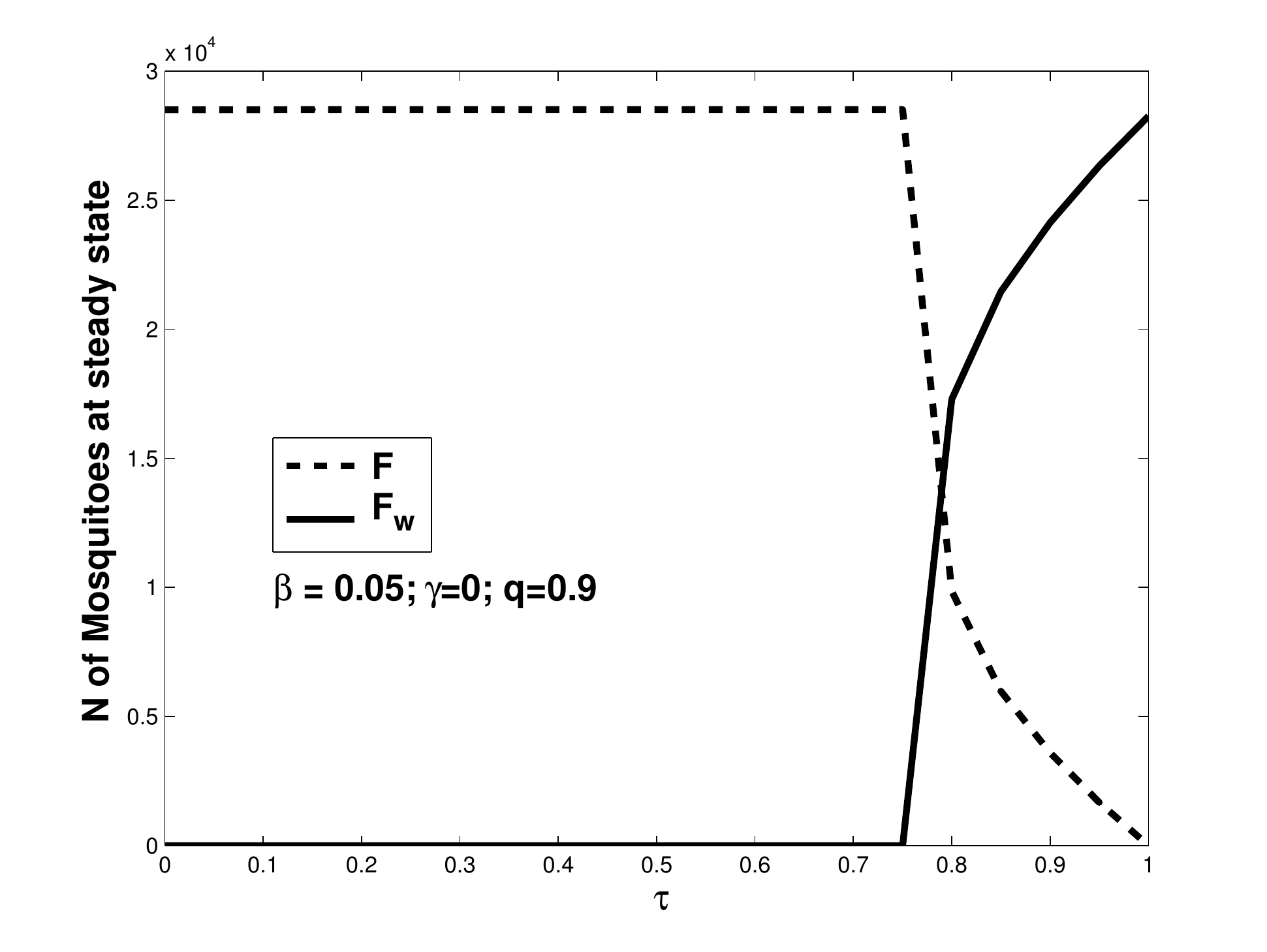}}
\subfigure{\includegraphics[width=6.6cm]{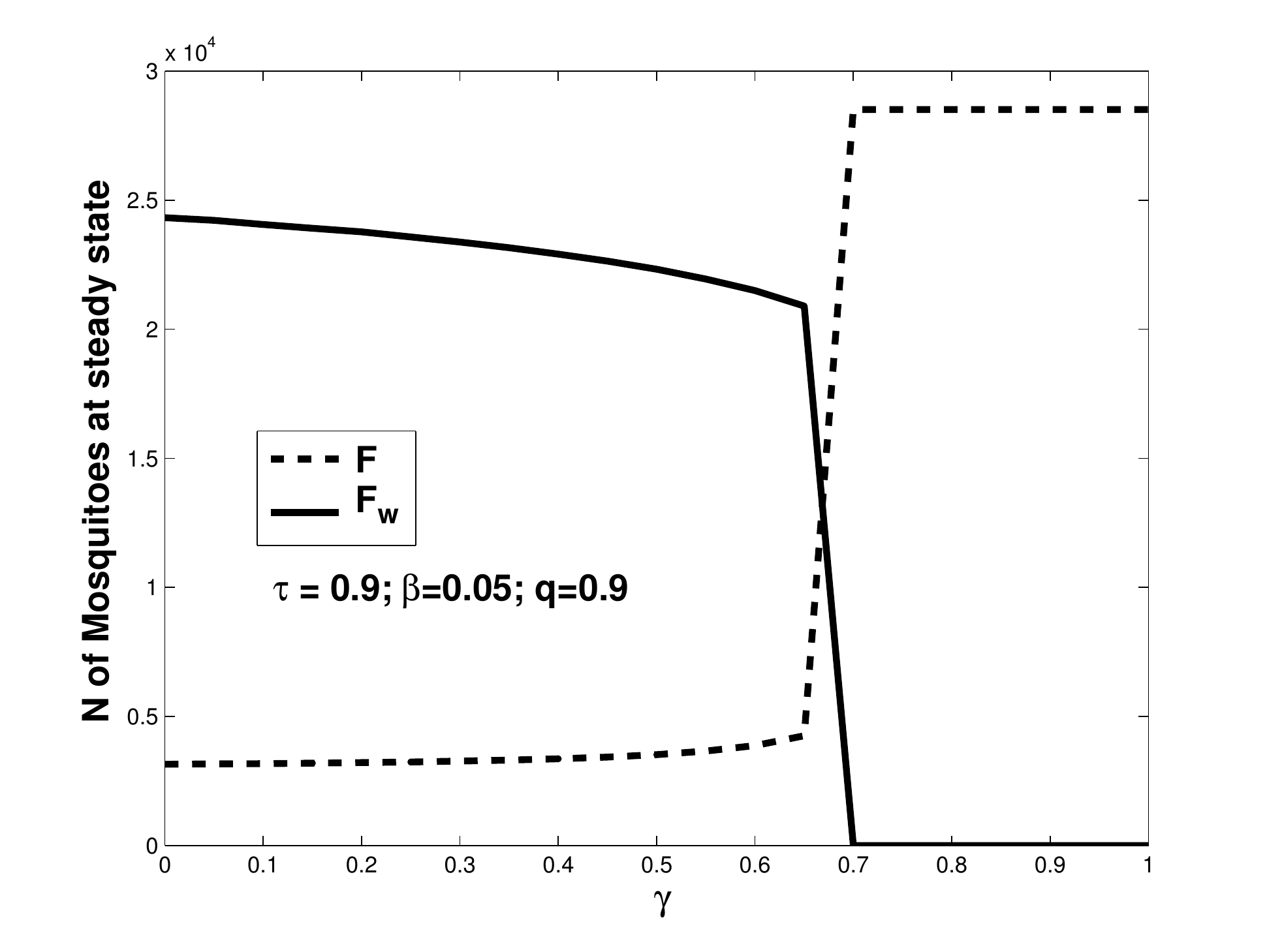}}
\subfigure{\includegraphics[width=6.6cm]{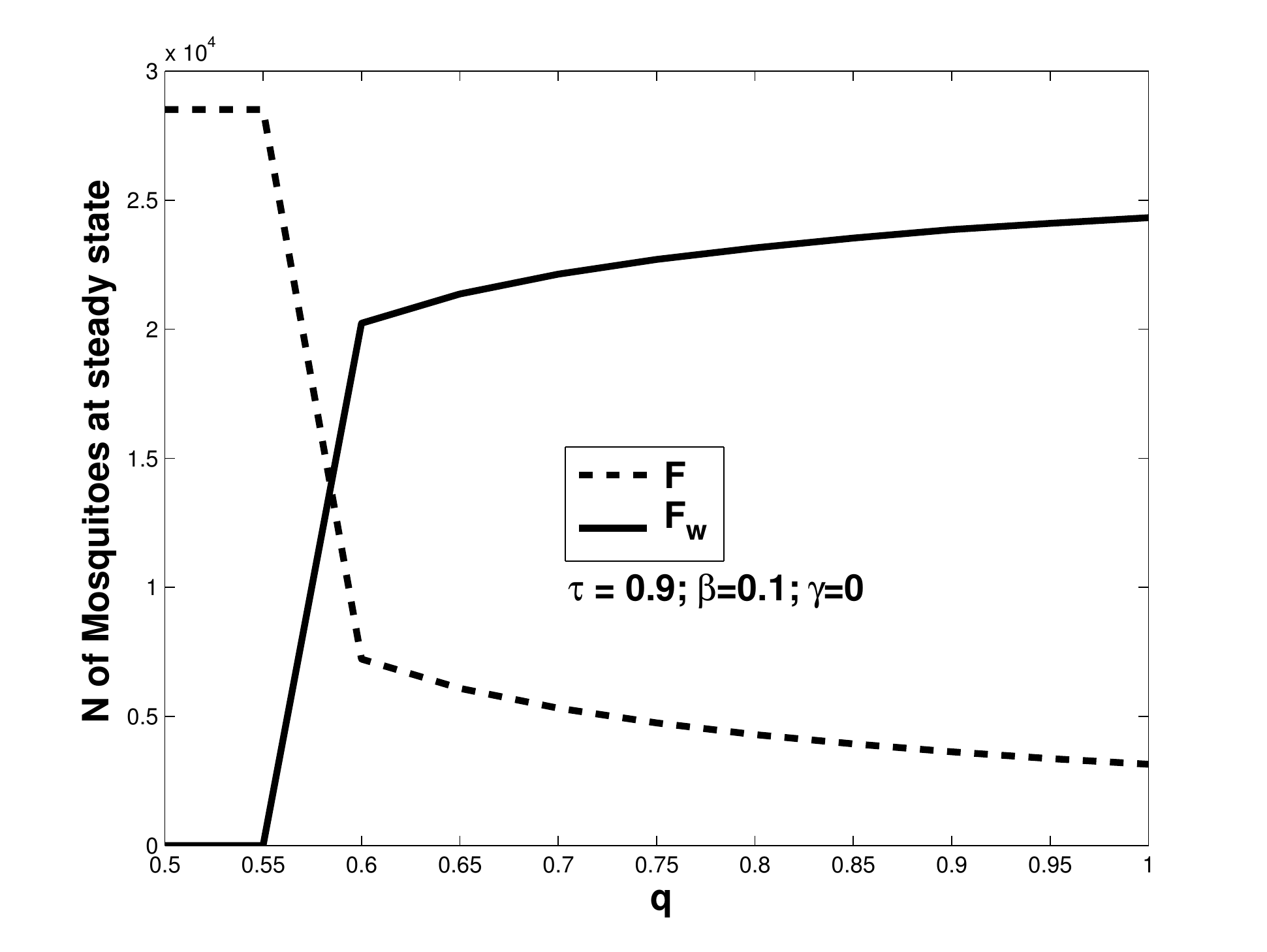}}
\caption{The dependence of the values of $F$ and $F_w$ at the steady
  state of model \eqref{eqM}-\eqref{eqFw} on the parameters $\beta,
  \tau, \gamma$ and $q$. Here, all the per-capita mortality rates of
  mosquitoes were taken as $0.05$, and we set $r=30$,
  $k=5000$. } \label{fig1} 
\end{figure}

\begin{figure}[H]
\centering
{\includegraphics[width=7cm]{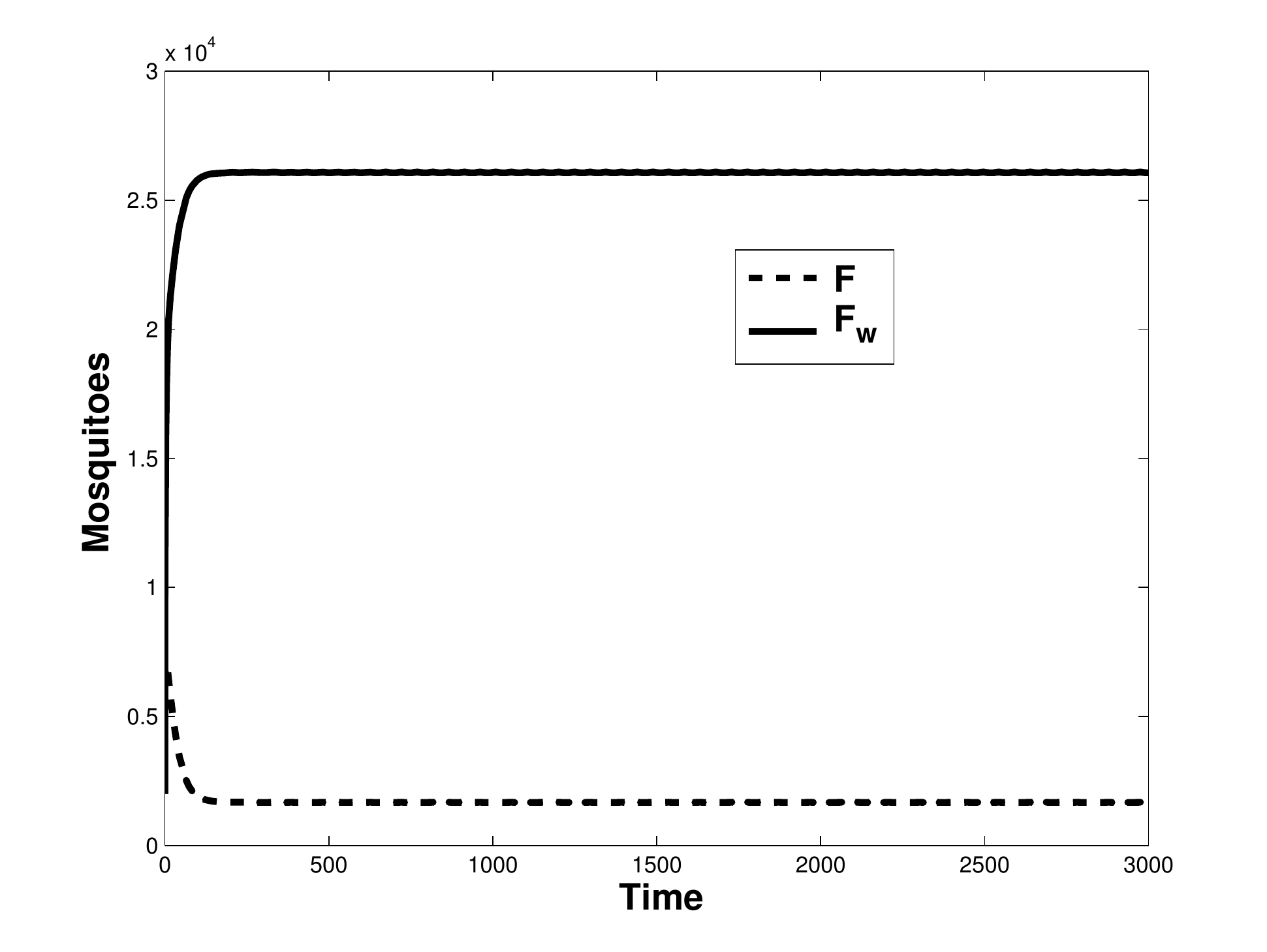}}
\caption{Simulation of model \eqref{eqM}-\eqref{eqFw} showing $F$ and $F_w$ against time. Here, all the per-capita mortality rates of mosquitoes are taken as $0.05$, and we set $r = 30$, $k = 5000$, $\beta = 0.1$, $q = 0.9$, $\gamma = 0$ and  $\tau = 0.95$. } \label{fig2}
\end{figure}

\begin{figure}[H]
\centering
{\includegraphics[width=7cm]{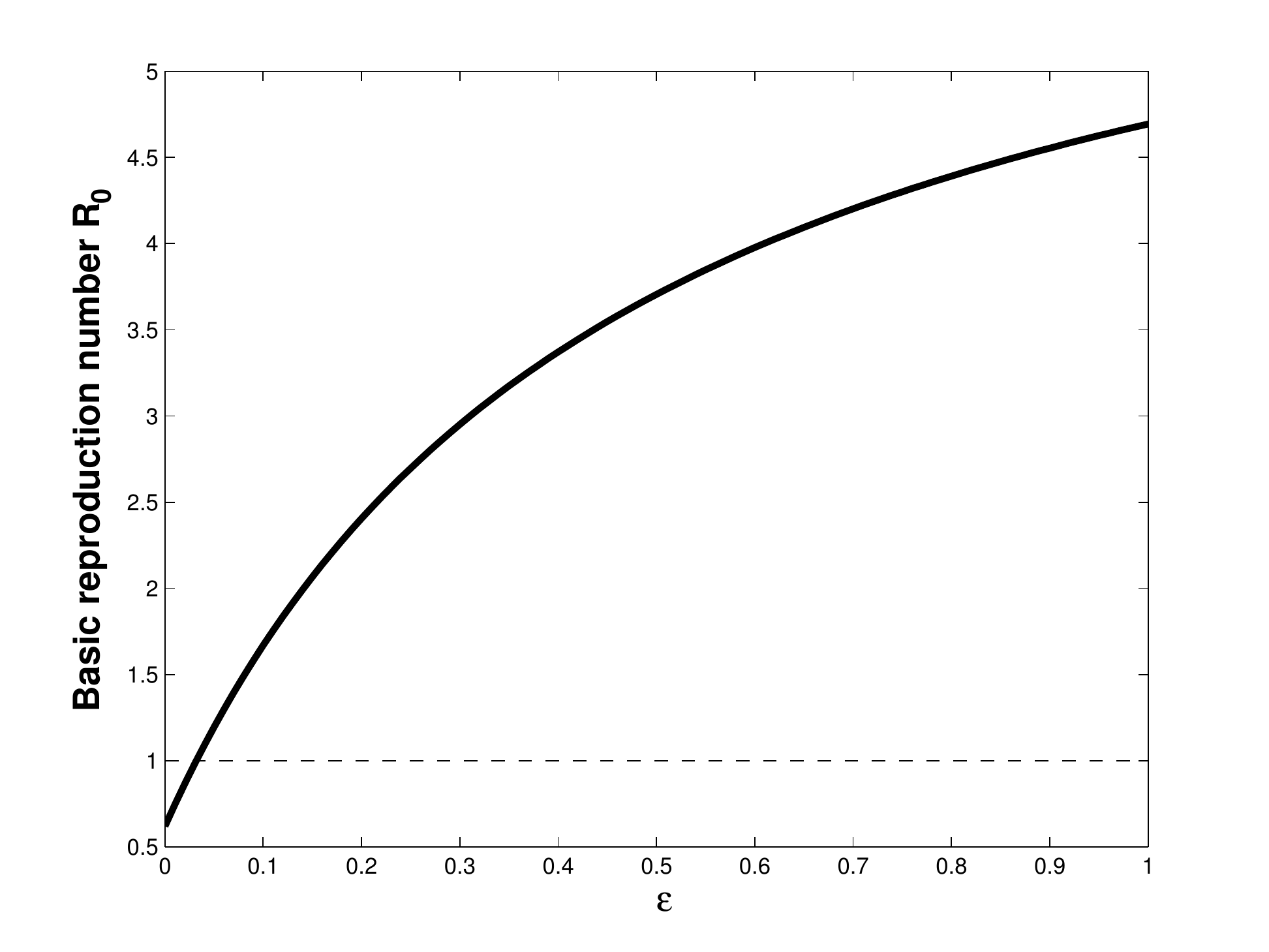}}
\caption{The basic reproduction number $R_0$, defined in \eqref{110914_1}, plotted against $\varepsilon$. The parameter values are $ \beta = 0.1$, $q = 0.9$, $\gamma = 0$, $\tau = 0.9$, with the other parameter values given in Table \ref{table1}. For these parameter values, $F_s^* = 3630$, $F_{ws}^* = 23837$, $B_s^* = 300$, and almost all of the mosquitoes are infected with {\it Wolbachia}. In this case, if $\varepsilon<0.03$, the basic reproduction number $R_0<1$, and the WNv will die out.} \label{fig3}
\end{figure}

\begin{figure}[H]
\centering
\subfigure[$B_s$ versus time]{\includegraphics[width=6.6cm]{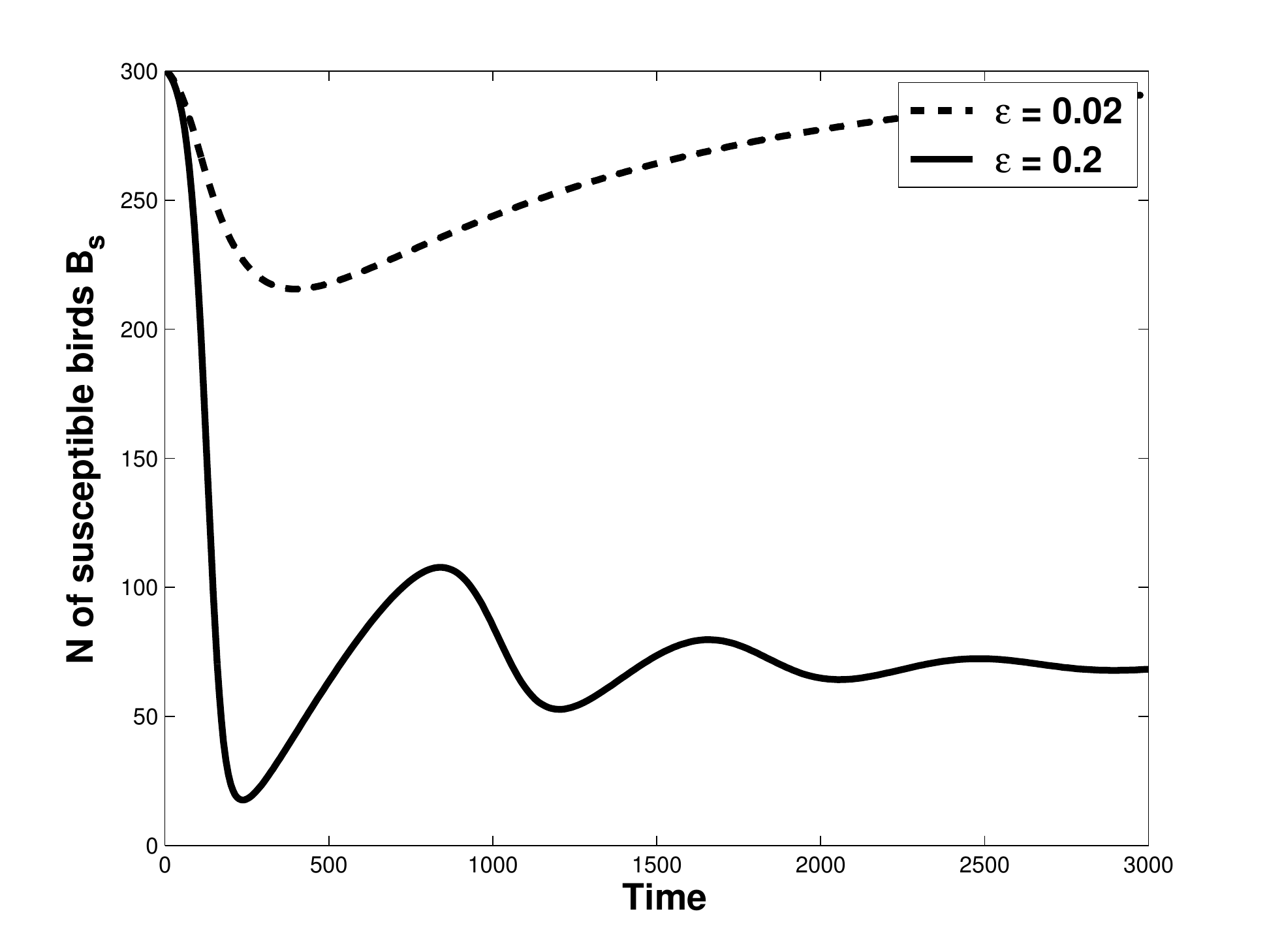}}
\subfigure[$B_i$ versus time]{\includegraphics[width=6.6cm]{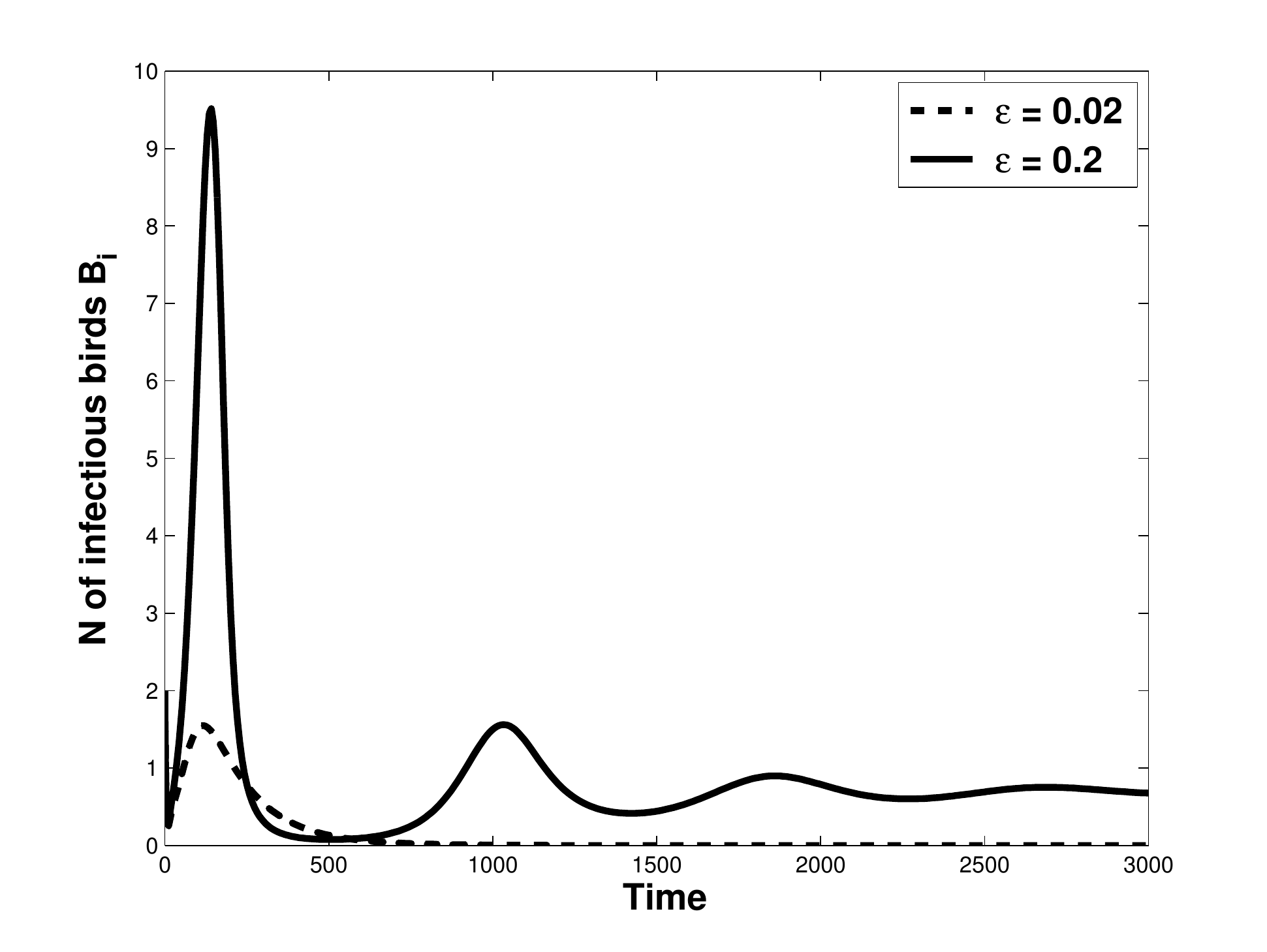}}
\subfigure[$F_i$ versus time]{\includegraphics[width=6.6cm]{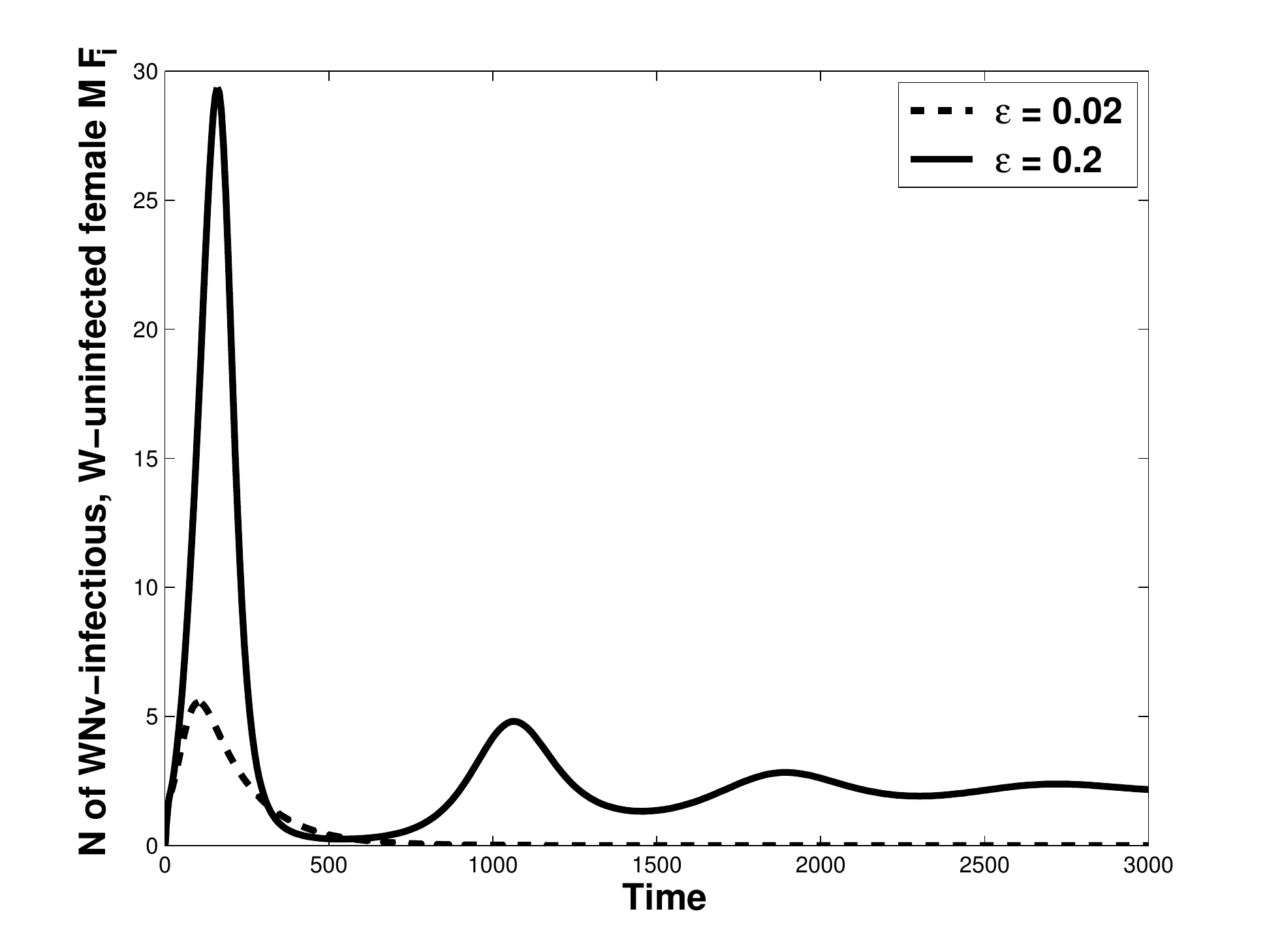}}
\subfigure[$F_{wi}$ versus time]{\includegraphics[width=6.6cm]{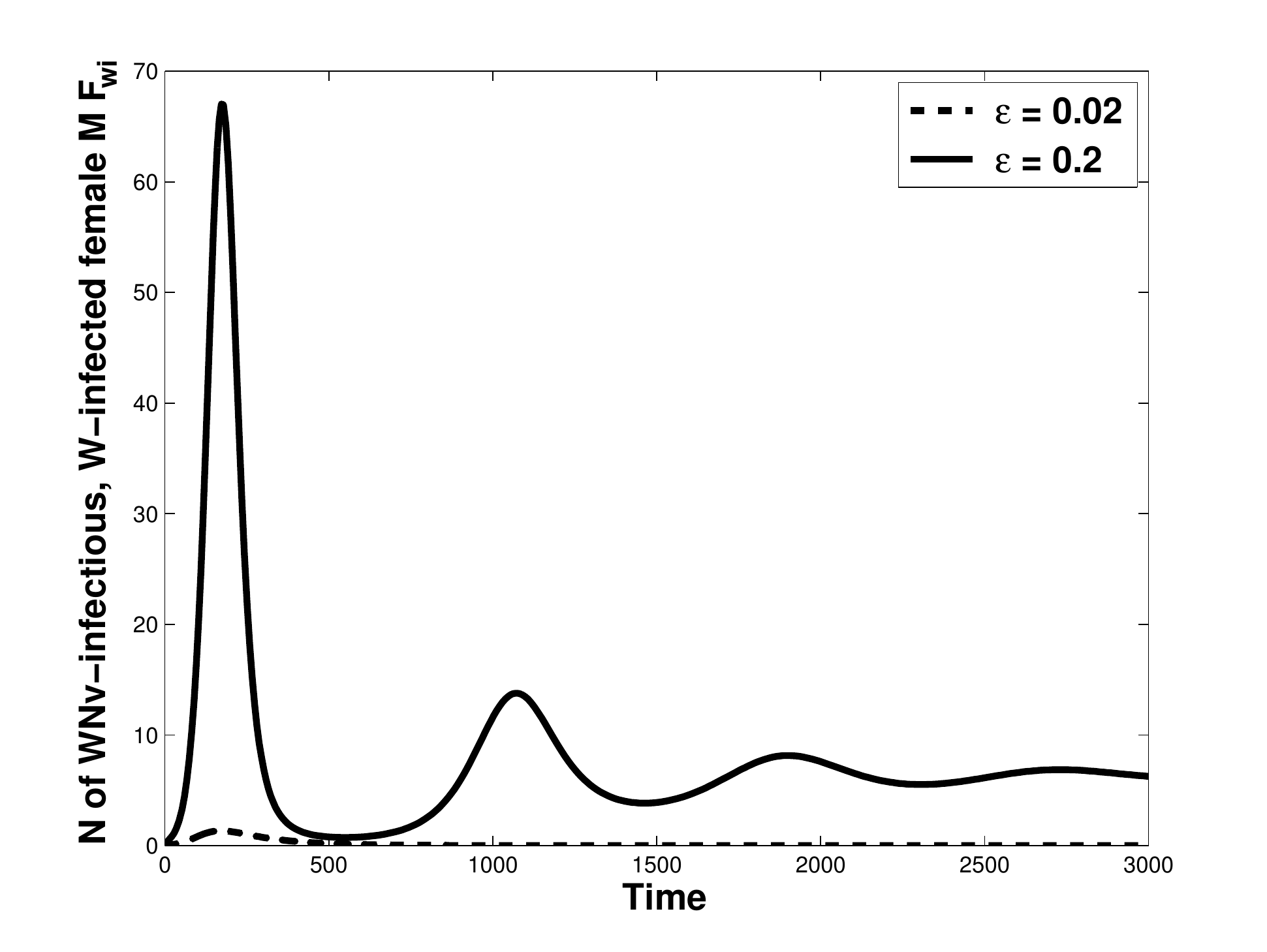}}
\caption{Simulation of model \eqref{F} with the parameter values as given in the caption of Figure \ref{fig3} for the cases $\varepsilon=0.2$ and $\varepsilon=0.02$. In the case $\varepsilon=0.02$, WNv dies out.}\label{fig4}
\end{figure}

\begin{figure}[H]
\centering
{\includegraphics[width=7cm]{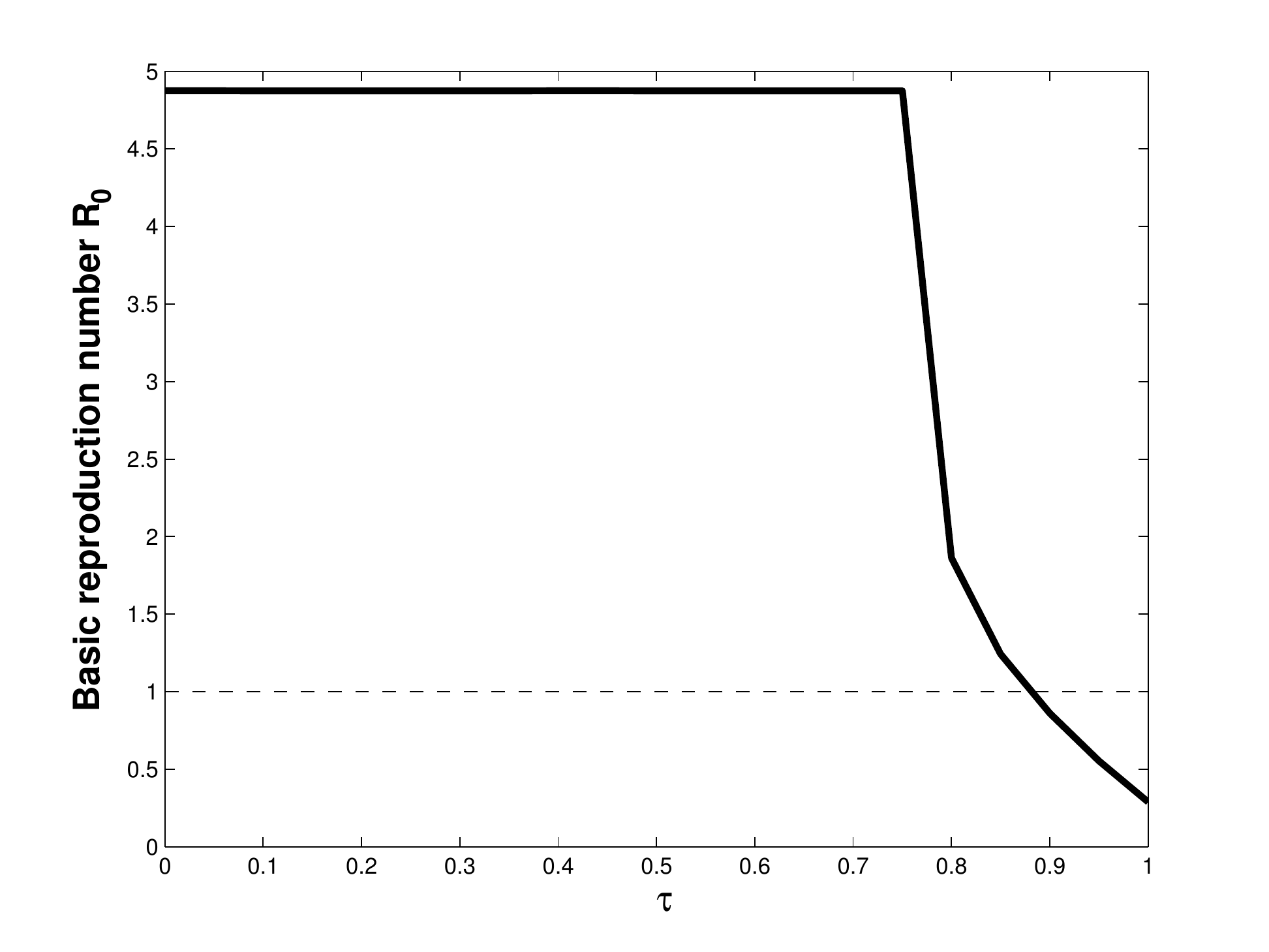}}
\caption{Basic reproduction number $R_0$, plotted against $\tau$, with $\beta = 0.1$, $q = 0.9$, $\gamma = 0$, $\varepsilon = 0.02$, and the other parameter values given in Table \ref{table1}. If $\tau>0.9$ the basic reproduction number $R_0<1$, and WNv will die out.} \label{fig5}
\end{figure}

\begin{figure}[H]
\centering
{\includegraphics[width=7cm]{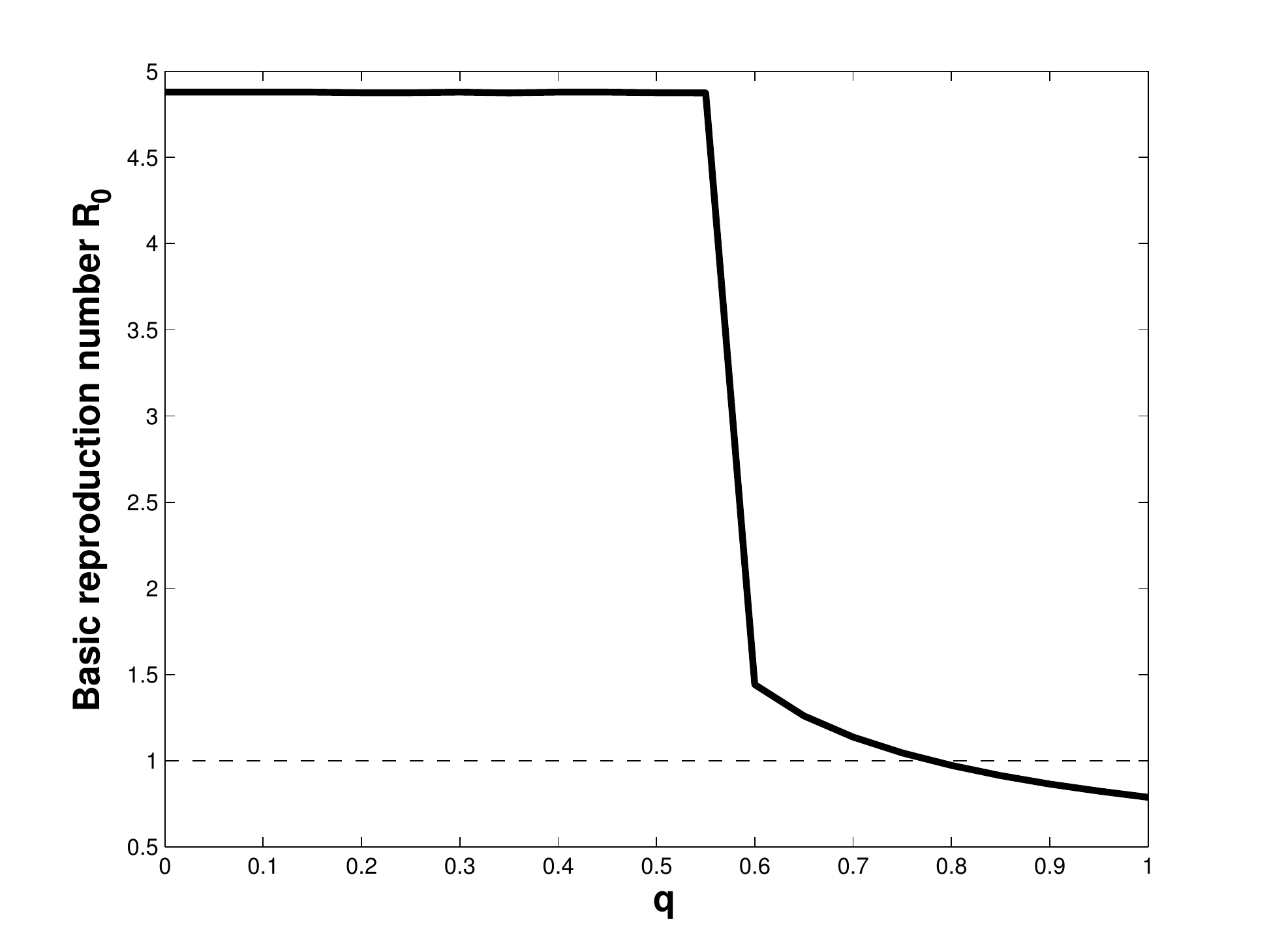}}
\caption{Basic reproduction number $R_0$, plotted against $q$, with $\beta = 0.1$, $\tau = 0.9$, $\gamma = 0$, $\varepsilon = 0.02$, and the other parameter values given in Table \ref{table1}. If $q>0.8$ the basic reproduction number $R_0<1$, and WNv will die out.} \label{fig6}
\end{figure}

\begin{figure}[H]
\centering
{\includegraphics[width=7cm]{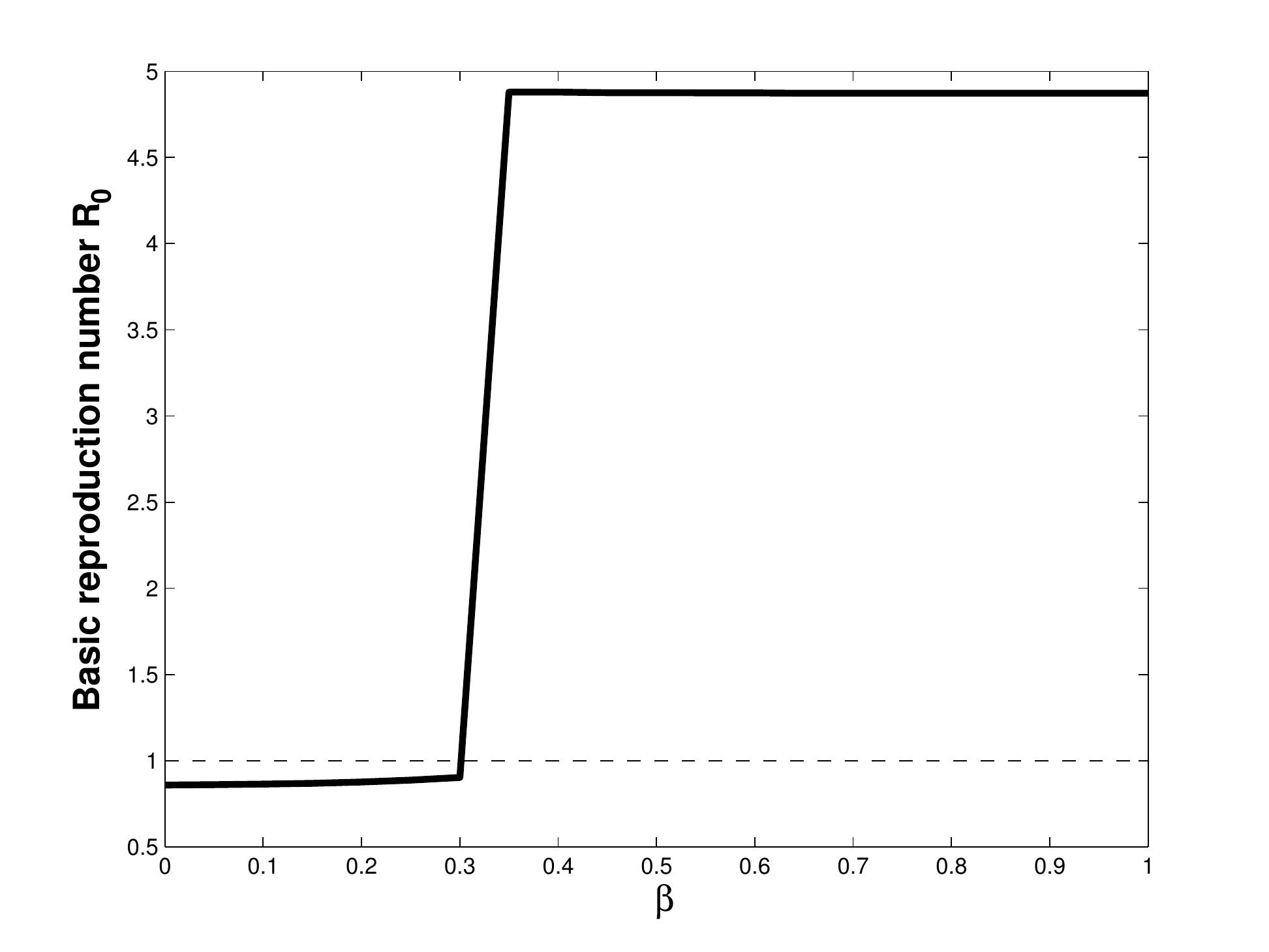}}
\caption{Basic reproduction number $R_0$, plotted against $\beta$, with $\tau = 0.9$, $q = 0.9$, $\gamma = 0$, $\varepsilon = 0.02$, and the other parameter values given in Table \ref{table1}. If $\beta<0.3$ the basic reproduction number $R_0<1$, and WNv will die out.} \label{fig7}
\end{figure}

\begin{figure}[H]
\centering
{\includegraphics[width=7cm]{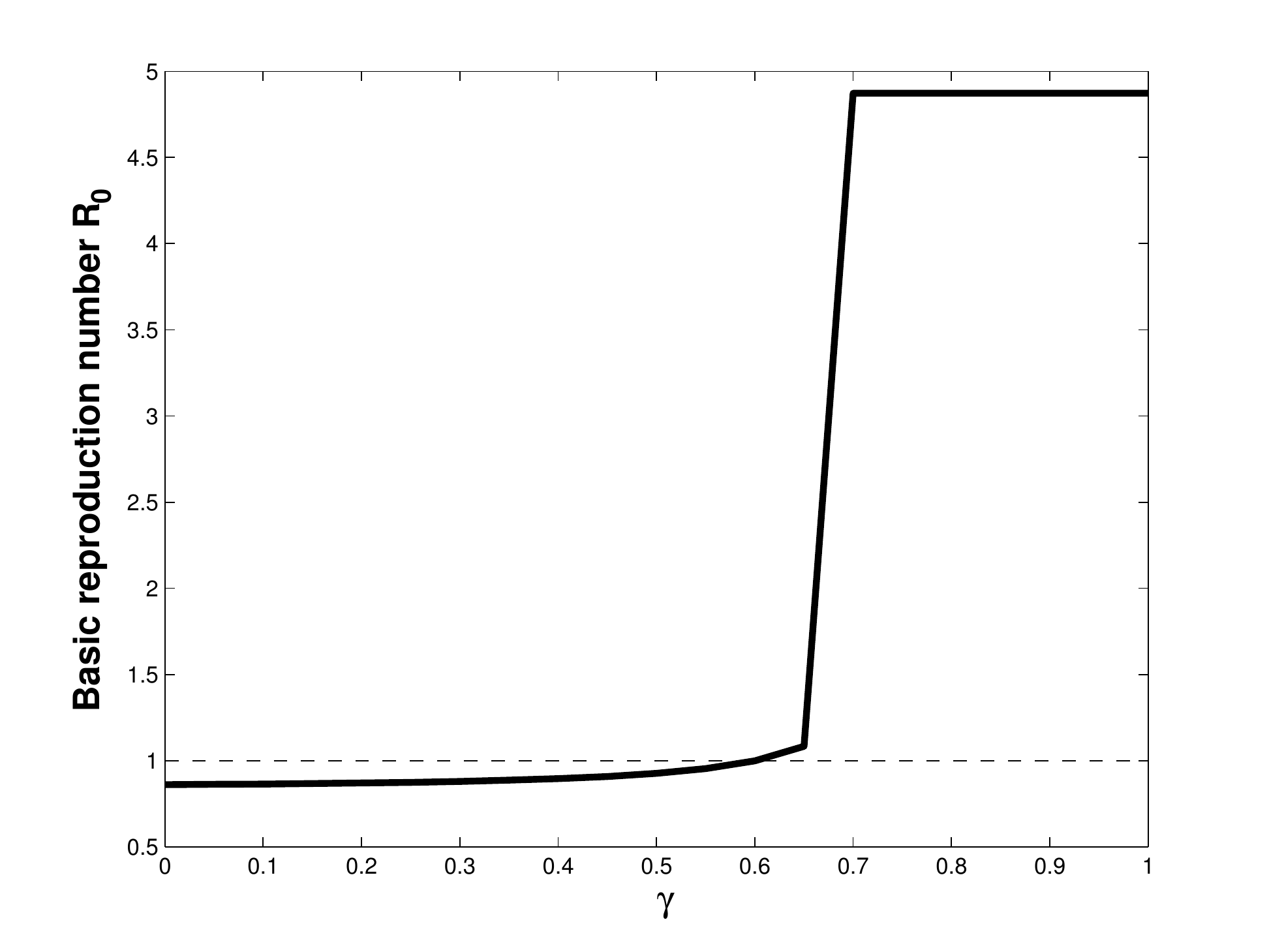}}
\caption{Basic reproduction number $R_0$, plotted against $\gamma$, with $\beta = 0.1$, $q = 0.9$, $\tau = 0.9$, $\varepsilon = 0.02$, and the other parameter values given in Table \ref{table1}. If $\gamma<0.65$ the basic reproduction number $R_0<1$, and WNv will die out.} \label{fig8}
\end{figure}

\section{Conclusion}
In this paper we have derived a detailed sex-structured model for a mosquito population infected with \Wolb. 
The model captures many of the well-known key effects of \Wolb infection, including cytoplasmic incompatibility, male killing, reduction in reproductive output and incomplete maternal transmission of the \Wolb infection. Our analysis shows that the mosquito population can stabilise at a \Wolb free equilibrium under certain circumstances, which include situations when inequality \eqref{A4} holds. Such circumstances include, for example, if \Wolb infection significantly reduces reproductive output, and/or \Wolb infection significantly lowers female life expectancy. 
We also showed that if $\tau=1$, i.e. maternal transmission of \Wolb is complete, then the mosquito population can stabilise at an equilibrium in which all mosquitoes are infected with \W. This happens in the case of sufficiently high cytoplasmic incompatibility. 
In the case of $\tau$ close to $1$ we have shown that \Wolb infected mosquitoes can coexist with small numbers of uninfected mosquitoes. 
We have also shown that under some additional assumptions our model has multiple coexistence steady states. 

We extended the sex-structured mosquito population model
\eqref{eqM}-\eqref{eqFw} to include West Nile virus, which is spread
by birds and mosquitoes, treating WNv as an SEI infection for
mosquitoes, and as an SEIR infection for birds. We were motivated by
results recently reported in \cite{Hussain2013}, which suggest that a
particular strain of \Wolb substantially reduces WNv replication in
the mosquito species {\it Aedes aegypti}. We modelled this crucial
phenomenon by incorporating a small parameter $\varepsilon$, the
reciprocal of which is proportional to the time spent in the WNv
exposed class for \Wolb infected mosquitoes. This enabled us to assess
the potential of \Wolb infection to eradicate WNv via its effect on
WNv replication in \Wolb infected mosquitoes. Notably the expression
we obtained for the basic reproduction number $R_0$ suggests that WNv
will be eradicated if at the steady state the overwhelming majority of
mosquitoes are infected with \W, and the \Wolb infection substantially
reduces WNv replication in mosquitoes. The first of these hypotheses
is in fact shown to hold for a number of \Wolb strains and mosquito
species, see e.g. \cite{Engelstadter2009}.

\section*{Acknowledgments}
We thank the American Institute of Mathematics for financial support
through the SQuaRE program. We also thank the reviewers for their helpful comments.

\end{document}